\documentclass[12pt,a4paper]{elsarticle}

\usepackage[table,xcdraw,svgnames]{xcolor}

\usepackage{amsmath, amscd, amsthm, amsfonts, amssymb, graphicx, subcaption, color, soul, enumerate,  mathrsfs, latexsym, bigstrut, framed, caption,txfonts}
\usepackage{pstricks,pst-node,pst-tree}

\usepackage[bookmarksnumbered, colorlinks, plainpages, backref]{hyperref}

\AtBeginDocument{%
	\hypersetup{colorlinks=true, linkcolor=black, anchorcolor=green, citecolor=midblue, urlcolor=blue, filecolor=magenta, pdftoolbar=true}}

\usepackage{algorithm}
\usepackage{algpseudocode}

\usepackage{etoolbox}

\usepackage{setspace}

\usepackage{tikz}
\usetikzlibrary{calc,decorations.pathreplacing, decorations.markings, positioning, shapes, topaths}

\usepackage[normalem]{ulem}
\definecolor{calpolypomonagreen}{rgb}{0.12, 0.3, 0.17}
\definecolor{byzantine}{rgb}{0.64, 0.14, 0.54}
\definecolor{cadmiumgreen}{rgb}{0.0, 0.42, 0.24}
\definecolor{caribbeangreen}{rgb}{0.0, 0.6, 0.46}

\newcommand{\Prob}{\ensuremath{\operatorname{Pr}}}

\newcommand{\D}{\ensuremath{\operatorname{DEG}}}

\definecolor{cupgreen}{rgb}{0,0.498,0.208}
\definecolor{cupblue}{rgb}{0,0,.5}
\definecolor{cupred}{rgb}{1,0.04,0}
\definecolor{cuppink}{rgb}{0.925,0,0.545}
\definecolor{cupmagenta}{rgb}{0.624,0.161,0.424}
\definecolor{cupbrown}{rgb}{0.71,0.212,0.133}

\definecolor{midblue}{rgb}{0.00,0.0,0.80}
\definecolor{darkblue}{rgb}{0.00,0.00,0.45}
\definecolor{thm}{rgb}{0.8,0,0.1}
\definecolor{sec}{rgb}{0,0,1}

\makeatother
\normalsize
\newtheorem{theorem}{{Theorem}}[section]

\theoremstyle{definition}
\newtheorem{definition}[theorem]{Definition}

\journal{Computers and Operations Research}

\begin{document}
	\begin{frontmatter}
		\title{Soft happy colourings and community structure of networks}
		
		\author[inst1]{Mohammad H. Shekarriz}
		\ead{m.shekarriz@deakin.edu.au}
		\affiliation[inst1]{
			organization={School of Information Technology, Deakin University},
			city={Geelong},
			state={VIC},
			country={Australia}}
		\author[inst1]{Dhananjay Thiruvady}
		\ead{dhananjay.thiruvady@deakin.edu.au}
		\author[inst1]{Asef Nazari}
		\ead{asef.nazari@deakin.edu.au}
		\author[inst2]{Rhyd Lewis}
		\ead{lewisr9@cardiff.ac.uk}
		\affiliation[inst2]{
			organization={School of Mathematics, Cardiff University},
			city={Cardiff},
			state={Wales},
			country={United Kingdom}
		}

		\begin{abstract}
			For $0<\rho\leq 1$, a $\rho$-happy vertex $v$ in a coloured graph $G$ has at least $\rho\cdot \deg(v)$ same-colour neighbours, and a $\rho$-happy colouring (aka soft happy colouring) of $G$ is a vertex colouring that makes all the vertices $\rho$-happy. A community is a subgraph whose vertices are more adjacent to themselves than the rest of the vertices. Graphs with community structures can be modelled by random graph models such as the stochastic block model (SBM). In this paper, we present several theorems showing that both of these notions are related, with numerous real-world applications. We show that, with high probability, communities of graphs in the stochastic block model induce $\rho$-happy colouring on all vertices if certain conditions on the model parameters are satisfied. Moreover, a probabilistic threshold on $\rho$ is derived so that communities of a graph in the SBM induce a $\rho$-happy colouring. Furthermore, the asymptotic behaviour of $\rho$-happy colouring induced by the graph's communities is discussed when $\rho$ is less than a threshold. We develop heuristic polynomial-time algorithms for soft happy colouring that often correlate with the graphs' community structure. Finally, we present an experimental evaluation to compare the performance of the proposed algorithms thereby demonstrating the validity of the theoretical results. 
		\end{abstract}
		
		
	\end{frontmatter}
	
	\section{Introduction}\label{Sec:Intro}
	Graph colouring is a powerful technique for solving complex network-related problems, with applications varying from engineering to quantum physics and social science to cybersecurity. Depending on the problem's nature, various types of graph colouring have already been introduced. From the conventional \emph{proper vertex colouring}~\cite{Colouring_book} to the recently defined \emph{happy colouring}~\cite{ZHANG2015117} (which is the subject of this paper), each of these versions appeared to be seminal in both theoretical and practical research.
	
	Happy colouring in graphs is proposed for detecting \emph{homophily}~\cite{Homophily} in social networks. For a $k$-coloured graph, Zhang and Li~\cite{ZHANG2015117} defined a vertex to be \emph{happy} if all its neighbours have the same colour as its colour, and an edge is \emph{happy} if both of its ends have the same colour. It is well-known that homophily in social networks can be expressed by their \emph{community structures} (see for example ~\cite{Fani2017, Solomon2019, Ye2020}), a structure that real-world networks often possess. In other words, homophily can be translated to grouping vertices into some sets in which they are connected with many edges to themselves and substantially fewer edges to the rest of the vertices. Each group is called a \emph{community}, aka a \emph{cluster}~\cite{10.1007/978-3-540-48413-4_23}. A community is then a subgraph whose vertices have most of their neighbours inside the community, and the problem has some similarities with finding the \emph{minimum cut} in the network~\cite{Newman2012}. Hence, by a community, we mean a group of vertices having an ample number of vertices that are adjacent together more than they are adjacent to the vertices of other communities.  
	This way, singletons are not considered as a community, while they are valid vertex cuts. 
	
	Zhang and Li in~\cite{ZHANG2015117} devised algorithms for maximising the number of happy vertices (MHV or $k$-MHV) 
	or maximising the number of happy edges (MHE or $k$-MHE) in a graph with some vertices already precoloured with $k$ colours. Additionally, they proved that MHV and MHE problems are solvable in polynomial time if $k\leq 2$, and they are NP-hard otherwise. Moreover, for $0<\rho\leq 1$, they introduced $\rho$-happy colouring as follows: a vertex $v$ is $\rho$-happy if its colour is the same as the colour of at least $\rho\cdot \deg(v)$ of its neighbours. The problem of maximising the number of $\rho$-happy vertices is called \emph{soft happy colouring}. Furthermore, they devised two algorithms for this problem, namely \hyperref[greedy1]{\sf Greedy-SoftMHV} and \hyperref[growth]{\sf Growth-SoftMHV}. 
	
	To the best of our knowledge, the problem of soft happy colourings has remained dormant. We focus on this problem 
	and our motivation to revisit this problem is its relation to the analysis of community structure in networks. However, conventional happy colouring (which is a special case of soft happy colouring when $\rho=1$) 
	and its two related problems, MHV and MHE, have already been the subject of several studies. For example, Aravind, Kalyanasundaram, and Kare~\cite{trees-happy} proposed linear time algorithms for MHV and MHE for trees. Meanwhile, Agrawal et al.~\cite{AGRAWAL202058} studied the complexity of (weighted) MHV and MHE based on the density measures of a graph.
	
	A few studies have proposed exact, heuristic, metaheuristic and matheuristic approaches to tackle the (conventional) MHV probem~\cite{Lewis2019265, thiruvady2020, THIRUVADY2022101188}.
	Lewis et al.~\cite{Lewis2019265} proved bounds on the number of happy vertices and introduced a \emph{construct, merge, solve, and adapt} algorithm --- CMSA for short --- for the MHV problem. Another study~\cite{thiruvady2020}, devised a \emph{tabu search} metaheuristic which proves superior to other available algorithms, especially on very large networks. The later study by Lewis et al.~\cite{LEWIS2021105114} explores a variant of the MHV problem, specifically the maximum induced happy subgraph in a network. They proposed some heuristic and metaheuristic algorithms and proved bounds on the maximum number of happy vertices.  Thiruvady and Lewis~\cite{THIRUVADY2022101188} continued this direction of study by developing approaches based on tabu search and evolutionary algorithms hybrid CMSA-tabu search --- CMSA-TS for short --- and showed that CMSA-TS is almost always the most effective approach for solving the MHV problem among the methods they tried. More recently, Ghirardi and Salassa~\cite{Ghirardi2022-lb} presented an algorithm to improve the quality of a solution given by any another algorithm. Zhao, Yan and Zhang~\cite{ZHAO2023113846} considered the relation of $k$-MHV and $k$-MHE with the \emph{maximum $k$-uncut problem} and proposed an algorithm with the time complexity of $\mathcal{O}(kn^2)$.
	
	Finding a connection between community structure and soft happy colouring is straightforward. If the communities are non-overlapping (i.e., communities partition the vertex set), we can express them by colour classes, which is the colouring \emph{induced by the graph's communities}. Consequently, this paper is about how communities of a graph can be considered as the $\rho$-happy colour classes. While the conventional happy colouring (i.e. $\rho=1$) with more than one colour can never be achieved for a connected graph~\cite{Lewis2019265}, we show that communities can make the entire graph $\rho$-happy when $0<\rho<1$. Figure~\ref{fig:community} illustrates this for a graph having 3 communities. As we see, here the communities in the illustrated graph entirely match the $0.5$-happy colour classes. 
	On the other hand, if a complete $\rho$-happy colouring is known for a graph $G$, then one may take its colour classes as an indication of the graph's communities.
	
	\begin{figure}
		\begin{center}
			\begin{tikzpicture}  [scale=0.8]
				
				\tikzstyle{every path}=[line width=1pt]
				
				\newdimen\ms
				\ms=0.1cm
				\tikzstyle{s1}=[color=black,fill,rectangle,inner sep=3]
				\tikzstyle{c1}=[color=black,fill,circle,inner sep={\ms/8},minimum size=2*\ms]

				\coordinate (a1) at  (-5,3.5);
				\coordinate (a0) at  (-4,2.5);
				\coordinate (a3) at (-4,4.5);
				\coordinate (a2) at (-3,3.5);
				\coordinate (a4) at  (4,4);
				\coordinate (a5) at (5,4);
				\coordinate (a6) at (5,3);
				\coordinate (a7) at  (3,3.5);
				\coordinate (a8) at (4.5,2);
				\coordinate (a9) at (1,-2);
				\coordinate (a10) at (1,-3);
				\coordinate (a11) at (-1,-3);
				\coordinate (a12) at (-1,-2);
				\coordinate (a13) at (0,-1);

				% draw edges
				
				\draw [color=black] (a0) -- (a1);
				\draw [color=black] (a0) -- (a2);
				\draw [color=black] (a1) -- (a2);
				\draw [color=black] (a1) -- (a3);
				\draw [color=black] (a2) -- (a3);
				
				\draw [color=black] (a4) -- (a5);
				\draw [color=black] (a4) -- (a6);
				\draw [color=black] (a4) -- (a7);
				\draw [color=black] (a4) -- (a8);
				\draw [color=black] (a5) -- (a6);
				\draw [color=black] (a6) -- (a7);
				\draw [color=black] (a6) -- (a8);
				
				\draw [color=black] (a9) -- (a10);
				\draw [color=black] (a9) -- (a11);
				\draw [color=black] (a9) -- (a13);
				\draw [color=black] (a10) -- (a11);
				\draw [color=black] (a10) -- (a12);
				\draw [color=black] (a11) -- (a12);
				\draw [color=black] (a12) -- (a13);
				
				\draw [color=black, dashed] (a0) -- (a11);
				\draw [color=black, dashed] (a2) -- (a7);
				\draw [color=black, dashed] (a3) -- (a4);
				\draw [color=black, dashed] (a4) -- (a13);
				\draw [color=black, dashed] (a2) -- (a12);
				\draw [color=black, dashed] (a5) -- (a9);
				\draw [color=black, dashed] (a8) -- (a10);

				% draw vertices
				\draw (a0) coordinate[c1, color=red];
				\draw (a1) coordinate[c1, color=red];
				\draw (a2) coordinate[c1, color=red];
				\draw (a3) coordinate[c1, color=red];
				\draw (a4) coordinate[c1, color=blue];
				\draw (a5) coordinate[c1, color=blue];
				\draw (a6) coordinate[c1, color=blue];
				\draw (a7) coordinate[c1, color=blue];
				\draw (a8) coordinate[c1, color=blue];
				\draw (a9) coordinate[c1, color=green];
				\draw (a10) coordinate[c1, color=green];
				\draw (a11) coordinate[c1, color=green];
				\draw (a12) coordinate[c1, color=green];
				\draw (a13) coordinate[c1, color=green];

			\end{tikzpicture}
			
		\end{center}
		\caption{\label{fig:community}
			A graph with 3 communities, represented by three colour classes. Dashed lines represent inter-community edges. Here the colouring is $0.5$-happy. }
	\end{figure}
	
	Section~\ref{sec:nomen} presents the required background and terminology for the theoretical contribution of this paper. Specifically, we review the stochastic block model, one of the simplest random graph models for graphs with community structure. It should also be emphasised that although we build our theorems and experimental tests based on this model, the relation between community structure and soft happy colouring can be applied to any other community structure model.
	
	The theoretical contribution of this paper is presented in Section~\ref{sec:soft_SBM}. Theorem~\ref{th:SBM-happy} gives a sufficient condition on the model parameters of a graph $G$ in the SBM so that the probability of its communities inducing a $\rho$-happy colouring is greater than a function. This enables us to find a threshold $0\leq\xi\leq 1$ such that, with high probability, communities of $G$ induce a $\rho$-happy colouring for $\rho\leq \xi$. Moreover, the asymptotic property of the induced colouring is detected in Theorem~\ref{th:infinity} as the probability of $\rho$-happines of this colouring approaches 1 when the number of vertices tends to infinity. The validity of these results is also discussed by experimental analysis in Section~\ref{sec:soft_SBM} as well as in Section \ref{Sec:exprimental}.
	
	In Section~\ref{sec:alg}, four heuristic algorithms are discussed for finding $\rho$-happy colouring. One of them, \hyperref[greedy3]{\sf Local Maximal Colouring}, is new and has shown a high correlation with the graphs' community structure. We report experimental results considering communities inducing colour classes in Section~\ref{Sec:exprimental} and carry out experiments on four algorithms for several randomly generated graphs.

	\section{Nomenclature}\label{sec:nomen}
	In this paper we use ``graphs" and ``networks" interchangeably, meaning simple unweighted finite graphs. Usually, a graph is denoted by its vertex set 
	and its edge set, i.e. $G=\left( V(G), E(G)\right)$. When two vertices $u$ and $v$ are adjacent, we write $u\sim v$ or $uv\in E(G)$. For $k\in\mathbb{N}$, a $k$-colouring of a graph $G$ is a function $c:V(G)\longrightarrow \{1,\ldots, k\}$. It is well known that a $k$-colouring partitions the vertex set into $k$ parts $V_1, \ldots, V_k$ while conversely, every $k$-partition of the vertex set induces a $k$-colouring which is unique up to permutation of colours. So, when colour orderings are unimportant for us, we write $c=\{V_1, \ldots, V_k\}$. A partial $k$-colouring is a function $c:S\longrightarrow \{1,\ldots, k\}$ where $S\subset V(G)$. By assigning another colour, say $k+1$, to the uncoloured vertices, every partial $k$-colouring becomes a $(k+1)$-colouring.
	
	For a vertex $v$, the set of its neighbours is denoted by $N(v)$. 
	The degree of a vertex $v$, denoted by $\deg(v)$, is the number of edges incident on it, i.e. $\deg(v)=|N(v)|$. The minimum and maximum degrees are respectively denoted by $\delta$ and $\Delta$. When we speak about a vertex $v$ in a set, a colour class, or a community of vertices, we define the degree of $v$ inside that set, which we denote here by $\deg_{in} (v)$. Further, standard graph theoretical notations and definitions can be found in~\cite{diestel2017} by Diestel. 
	
	\subsection{Happy colouring}\label{sec:classifying}
	
	We now formally define a $\rho$-happy colouring of a graph. It should be said that these definitions have their origins in the very first paper on the subject of happy colourings~\cite{ZHANG2015117}. 
	
	\begin{definition}
		\label{def:rho-happy}
		Let $G$ be a graph, $0\leq \rho\leq1$ and $c:V(G)\longrightarrow \{1,\ldots,k\}$ for $k\in\mathbb{N}$ be a vertex colouring. A vertex $v\in V(G)$ is called \emph{$\rho$-happy} if at least $\rho\cdot\deg(v)$ of its neighbours have the same colour as $c(v)$.  
		The colouring $c$ is called $\rho$-happy if every vertex of $G$ is $\rho$-happy.
	\end{definition}

	Immediately, we can see that 1-happy colouring is the conventional happy colouring, while every vertex colouring is a $0$-happy colouring. For practical purposes, we may assume $\rho\neq 0$. Moreover, because the number of vertices having the same colour 
	in any neighbourhood is integral, a vertex $v$ being $\rho$-happy is equivalent to having at least $\lceil \rho\cdot\deg (v) \rceil$ same-colour neighbours, see Figure~\ref{fig:rho-happy}.

	\begin{figure}
		\begin{center}
			\begin{tabular}{ c c c }
				\begin{tikzpicture}  [scale=2]
					
					\tikzstyle{every path}=[line width=1pt]
					
					\newdimen\ms
					\ms=0.1cm
					\tikzstyle{s1}=[color=black,fill,rectangle,inner sep=3]
					\tikzstyle{c1}=[color=black,fill,circle,inner sep={\ms/8},minimum size=2*\ms]

					\coordinate (a0) at  (0,0);
					\coordinate (a1) at  (1,0);
					\coordinate (a2) at (0.7660,0.6427);
					\coordinate (a3) at (0.1736,0.9848);
					\coordinate (a4) at  (-0.4999,0.8660);
					\coordinate (a5) at (-0.9397,0.3421);
					\coordinate (a6) at (-0.9397,-0.3419);
					\coordinate (a8) at (0.1735,-0.9848);
					\coordinate (a9) at (0.7659,-0.6429);

					\draw [color=black] (a0) -- (a1);
					\draw [color=black] (a0) -- (a3);
					\draw [color=black] (a0) -- (a4);
					\draw [color=black] (a0) -- (a5);
					\draw [color=black] (a0) -- (a2);
					\draw [color=black] (a5) -- (a0);
					\draw [color=black] (a0) -- (a6);
					\draw [color=black] (a0) -- (a8);
					\draw [color=black] (a0) -- (a9);

					\draw (a1) coordinate[c1, color=red];
					\draw (a2) coordinate[c1];
					\draw (a3) coordinate[c1];
					\draw (a4) coordinate[c1];
					\draw (a5) coordinate[c1, color=red];
					\draw (a6) coordinate[c1];
					\draw (a8) coordinate[c1, color=blue];
					\draw (a9) coordinate[c1, color=red];
					\draw (a0) coordinate[c1, label=below left:$u$];

				\end{tikzpicture}
				& \hspace{6mm} &
				
				\begin{tikzpicture}  [scale=2]
					
					\tikzstyle{every path}=[line width=1pt]
					
					\newdimen\ms
					\ms=0.1cm
					\tikzstyle{s1}=[color=black,fill,rectangle,inner sep=3]
					\tikzstyle{c1}=[color=black,fill,circle,inner sep={\ms/8},minimum size=2*\ms]

					\coordinate (a0) at  (0,0);
					\coordinate (a1) at  (1,0);
					\coordinate (a2) at (0.7660,0.6427);
					\coordinate (a3) at (0.1736,0.9848);
					\coordinate (a4) at  (-0.4999,0.8660);
					\coordinate (a5) at (-0.9397,0.3421);
					\coordinate (a6) at (-0.9397,-0.3419);
					\coordinate (a7) at  (-0.5001,-0.8660);
					\coordinate (a8) at (0.1735,-0.9848);
					\coordinate (a9) at (0.7659,-0.6429);

					\draw [color=black] (a0) -- (a1);
					\draw [color=black] (a0) -- (a3);
					\draw [color=black] (a0) -- (a4);
					\draw [color=black] (a0) -- (a5);
					\draw [color=black] (a0) -- (a2);
					\draw [color=black] (a5) -- (a0);
					\draw [color=black] (a0) -- (a7);
					\draw [color=black] (a0) -- (a6);
					\draw [color=black] (a0) -- (a8);
					\draw [color=black] (a0) -- (a9);

					\draw (a1) coordinate[c1, color=red];
					\draw (a2) coordinate[c1];
					\draw (a3) coordinate[c1];
					\draw (a4) coordinate[c1];
					\draw (a5) coordinate[c1, color=red];
					\draw (a6) coordinate[c1];
					\draw (a7) coordinate[c1, color=green];
					\draw (a8) coordinate[c1, color=blue];
					\draw (a9) coordinate[c1, color=red];
					\draw (a0) coordinate[c1, label=below left:$v$];

				\end{tikzpicture}
				\\
				
				(a)& &(b)
			\end{tabular}
		\end{center}
		\caption{\label{fig:rho-happy} An example of a $\rho$-happy colouring problem using $\rho=0.5$. (a) The vertex $u$ is $0.5$-happy because its degree is 8 and at least $\lceil 0.5\times 8\rceil=4$ of its neighbours have the same colour as its colour (black). (b) The vertex $v$ is not $0.5$-happy because its degree is 9 and $\lceil 0.5\times 9\rceil= 5$, but only $4$ of its neighbours have the same colour as its colour (black).}
	\end{figure}

	\subsection{Stochastic block model}\label{sec:SBM}
	For modelling general networks on $n$ nodes, random graphs are often used. They are probability spaces consisting of $n$-vertex graphs. See~\cite{Bollobas_2001} by Bollob\'{a}s for further details on random graphs. There are several random graph models such as the \emph{Erd\H{o}s-R\'{e}nyi model}~\cite{erdos59a}, the \emph{hierarchical network model}~\cite{Ravasz-PhysRevE.67.026112}, the \emph{Watts–Strogatz model}~\cite{Watts1998}, and the \emph{stochastic block model} (SBM)~\cite{JERRUM1998155}. It should be noted that graphs with high probability in the Erd\H{o}s-R\'{e}ny model seldom have community structure. One reason is that in this model, the expected number of edges between any pair of vertices are equal. Consequently, this model is not suitable for studying graphs with community structure. In contrast, other mentioned models are designed so that their graphs, whose probabilities in their spaces are the highest, possess community structures. 
	
	One of the most popular random graph models for generating and modelling networks with community structure is the \emph{Stochastic block model} which is introduced by Holland, Laskey, and Leinhardt~\cite{HOLLAND1983109}. The model is also known as the \emph{planted partition model}~\cite{10.1007/978-3-540-48413-4_23, Condon2001}, and is considered in numerous research papers~\cite{Abbe2018Recent, Abbe2016, Abbe2015, JERRUM1998155, Lee2019, PhysRevLett.108.188701, Zhang_2012}.  \
	The model itself has also been subject to further generalisation such as the \emph{degree-corrected stochastic block model}~\cite{Karrer-PhysRevE.83.016107} and the \emph{Ball et al. model}~\cite{Ball-PhysRevE.84.036103}. In this paper, we limit ourselves to the conventional form of the stochastic block model for our theoretical results (see Section~\ref{sec:soft_SBM}) and experimental tests (see Section~\ref{Sec:exprimental}). We now summarise the main concepts and terminology.
	
	In the simplest form of the SBM, denoted by $\mathcal{G}(n,k,p,q)$, a graph has $n$ vertices and $k$ communities. Each vertex belongs to exactly one of the $k$ communities, hence the communities are not overlapping. Moreover, communities are also assumed to be of nearly the same size $\approx\frac{n}{k}$, where if $n$ is not divisible by $k$, some communities have only one additional vertex compared to the others. 
 
In the general case of the SBM, the probability of the existence of an edge $uv$ depends on community membership of $u$ and $v$. If $C_1, \ldots, C_k$ are the communities and $u\in C_i$ while $v\in C_j$ for $i,j\in \{1,\ldots, k\}$, then the probability of $uv\in E(G)$ is $P_{ij}$. In the simplest form of $\mathcal{G}(n,k,p,q)$, however, each $P_{ij}$ is considered to be either $p$ or $q$ for $0<q<p<1$. In other words, we assume 
	\begin{equation}P_{ij}=\begin{cases} p & \text{ if } i=j \\ q & \text{otherwise.} \end{cases}\end{equation}
	As a result, the expected degree of every vertex $v\in C_i$ is 
	\begin{equation}
		\overline{d}=\mathbb{E}[\D(v)]=\left(\frac{n}{k}-1\right)p+\frac{k-1}{k}nq,
	\end{equation}
where $\D(v)$ represents the random variable giving $\deg(v)$ for each $v\in V(G)$. When $n$ is sufficiently large, we can assume $\overline{d}\simeq \frac{n}{k}p+\frac{k-1}{k}nq$. 
	
	It is evident that when $p$ and $q$ are very close to each other, the SBM becomes indistinguishable from the Erd\H{o}s-R\'{e}ny model and we might be unable to identify any community structure in the graph. For $k=2$, Decelle et al.~\cite{PhysRevE.84.066106} have also
	conjectured that detecting communities in large networks in the SBM is possible if and only if $(a -b)^2 > 2(a +b)$ for $a=np$ and $b=nq$. It is shown afterwards that when $(a-b)^2 > C(a+b)$ for a large enough positive integer $C$, %\AN{more info on C} 
	the communities can be revealed by the spectral clustering algorithm~\cite{Coja-Oghlan}. The conjecture was finally proven in~\cite{Mossel2015}. For the general case (when $k\geq 2$), there is a --- distinguishability of community structure --- phase transition result worth noting~\cite{Abbe2015}. In any case, the larger the value for $\frac{p}{q}$, the more distinguishable the community structure. Therefore, we assume, especially in our tests, that $\frac{p}{q}$ is large enough so that the communities are distinguishable.
	Further information about recent developments on the SBM can be found in~\cite{Abbe2018Recent} by Abbe.

	\section{Soft happy colouring of graphs in the SBM}\label{sec:soft_SBM}
	
	For any pair of vertices $u$ and $v$ of a graph $G$ in the SBM, the adjacency probability is known. If $u$ and $v$ belong to the same community, the probability of $u$ being adjacent to $v$ is $p$ while if they belong to different communities, they are adjacent with the probability $q$. Using this fact, one might like to know that in the induced colouring by communities, what the probability $\rho$-happiness of a vertex is. In other words, given an $0<\varepsilon <1$, we are interested in finding a condition that the probability of $\rho$-happiness of a vertex is at least $1-\varepsilon$. Knowing this condition enables us to increase the probability of $\rho$-happiness of the induced colouring by the communities. The following theorem gives this condition.
	
	\begin{theorem}\label{th:SBM-happy}
		Suppose that $G$ is a graph modelled by the SBM, i.e., $G\in\mathcal{G}(n,k,p,q)$, $n$ is a sufficiently large integer, $2\leq k$, $0<q<p<1$, and $0<\rho\leq 1$. Then, for $0<\varepsilon<1$, at least with the probability of $(1-\varepsilon)^n$ the communities of $G$ induce a $\rho$-happy colouring on $G$ if %and only if 
		\begin{equation}\label{th_eq}
			q(k-1)(e^\rho -1) +p(e^\rho-e)<\frac{k}{n}\mathrm{ln}(\varepsilon).
		\end{equation}
	\end{theorem}
	\begin{proof}
		Let the communities of $G$ be $C_1,\ldots, C_k$ and let $v$ be an arbitrary vertex of the $j^\text{th}$ community of $G$. For $i\in\{1,\ldots, k\}$, let $X_i$ be the random variable which gives the number of adjacent vertices to $v$ in $C_i$. In other words,
		\begin{equation}
			X_i=\sum_{u\in C_i} A_{uv},
		\end{equation}
		where $A_{uv}\in\{0,1\}$ are independent random variables denoting the number of edges between $u$ and $v$. Then $X_i$ have a binomial distribution and $$\mathbb{E}[X_i]=\begin{cases}
			p\frac{n}{k}& \text{ if } i=j\\
			q\frac{n}{k}& \text{ otherwise.}
		\end{cases}$$
		For an arbitrary $t>0$, we can now easily find the moment generating function $$\mathbb{E}\left[ e^{t\cdot X_i}\right] =\begin{cases}
			\left( 1-p+p\cdot e^t\right)^\frac{n}{k}&\text{ if } i=j\\
			\left( 1-q+q\cdot e^t\right)^\frac{n}{k}&\text{ otherwise.}
		\end{cases}$$
		Note that $$\D(v)=\sum_{i=1}^k X_i$$ while $$\D_{in} (v)=X_j$$ are random variables giving $\deg(v)$ and $\deg_{in} (v)$, respectively.
		
		Now, by the Chernoff bound, for $t>0$ we have $$\Prob(\rho\D(v)-\D_{in}(v)>0)\leq \mathbb{E}\left[e^{t(\rho \D(v)-\D_{in}(v))}\right].$$
		Therefore, $$\Prob(\rho\D(v)-\D_{in}(v)>0)\leq   \left( 1-p+p\cdot e^{t\rho}\right)^\frac{n}{k}\cdot  \left( 1-p+p\cdot e^t\right)^{-\frac{n}{k}}\cdot \prod_{i=1}^{k-1}\left( 1-q+q\cdot e^{t\rho}\right)^\frac{n}{k},$$
		which means that \[\begin{split}
			\Prob(\rho\D(v)-\D_{in}(v)>0)&\leq   \left( 1-p+p\cdot e^{t\rho}\right)^\frac{n}{k}\cdot  \left( 1-p+p\cdot e^t\right)^{-\frac{n}{k}}\cdot \left( 1-q+q\cdot e^{t\rho}\right)^{\frac{k-1}{k}n}\\
			& \leq e^{\mathrm{ln}\left( \left( 1-p+p\cdot e^{t\rho}\right)^\frac{n}{k}\cdot  \left( 1-p+p\cdot e^t\right)^{-\frac{n}{k}}\cdot \left( 1-q+q\cdot e^{t\rho}\right)^{\frac{k-1}{k}n}\right) }\\
			&\leq e^{\frac{n}{k}\cdot\left( \mathrm{ln}\left(1-p+p\cdot e^{t\rho}\right)-\mathrm{ln}\left(1-p+p\cdot e^{t}\right)+(k-1)\cdot\mathrm{ln}\left(1-q+q\cdot e^{t\rho}\right)\right)}.
		\end{split}\]
		Since $e^x\geq 1+x$ for $x\in\mathbb{R}$, we can deduce that $$1-p+p\cdot e^{t\rho}=1+p\cdot\left(e^{t\rho}-1\right)\leq e^{p\cdot\left(e^{t\rho}-1\right)}.$$ Consequently, using similar inequalities, we have
		\[\begin{split}
			\Prob(\rho\D(v)-\D_{in}(v)>0)&\leq e^{\frac{n}{k}\cdot\left( p\cdot \left(e^{t\rho}-1\right)-p\cdot\left(e^{t}-1\right)+(k-1)q\cdot\left(e^{t\rho}-1\right)\right)}\\
			&\leq e^{\frac{n}{k}\cdot\left( p\cdot \left(e^{t\rho}-e^t\right)+(k-1)q\cdot\left(e^{t\rho}-1\right)\right)}.
		\end{split}\]
		Hence, for $\Prob(\rho\D(v)-\D_{in}(v)>0)<\varepsilon$, it is sufficient to have 
		\begin{equation}\label{epsilon}
			e^{\frac{n}{k}\cdot\left( p\cdot \left(e^{t\rho}-e^t\right)+(k-1)q\cdot\left(e^{t\rho}-1\right)\right)}<\varepsilon,
		\end{equation}
		or $$p\cdot \left(e^{t\rho}-e^t\right)+(k-1)q\cdot\left(e^{t\rho}-1\right)<\frac{k}{n}\mathrm{ln}(\varepsilon).$$ 
		The inequality also holds when $t=1$, resulting in the inequality we wanted to prove. 
		
		On the other hand, when Inequality~(\ref{th_eq}) holds, the probability of $v$ being $\rho$-unhappy is less than $\varepsilon$, and therefore, with the probability $1-\varepsilon$, the communities of graph $G$ induce a colouring on vertices of $G$ in which $v$ is $\rho$-happy. Since the values of $A_{uv}$ are independent for all pairs $u,v\in V(G)$, it can be inferred that $\D(v)$ and $\D_{in}(v)$ are independent from $\D(u)$ and $\D_{in}(u)$, respectively. 
		Consequently, $\rho\D(v)-\D_{in}(v)$ and $\rho\D(u)-\D_{in}(u)$ are independent too, which means that when Inequality~(\ref{th_eq}) holds, with the probability of at least $(1-\varepsilon)^n$, the communities of $G$ induce a $\rho$-happy colouring on $G$.
	\end{proof}
	
	We now need additional notation to consider the implications of Theorem~\ref{th:SBM-happy}. Suppose that $G$ is a graph modelled by the SBM, $v$ is a vertex, $0\leq \rho\leq 1$, and $c$ is the vertex colouring induced by the communities of $G$. 
	If $c$ makes $v$ $\rho$-happy, we write $v\in H_\rho$ and, if it makes the entire graph $\rho$-happy, we write $G\in H_\rho$. By Theorem~\ref{th:SBM-happy}, we know that $\Prob (v\in H_\rho )\geq 1-\varepsilon$, while $\Prob (G\in H_\rho )\geq (1-\varepsilon)^n$. Since this is dependent on the choice of $\varepsilon$, using Equation~(\ref{epsilon}) we define
	\begin{equation}\label{epsilon_bar}
		\tilde{\varepsilon}= e^{\frac{n}{k}\cdot\left( p\cdot \left(e^{\rho}-e\right)+q(k-1)\cdot\left(e^{\rho}-1\right)\right)}.
	\end{equation}
	Therefore, it can be seen that $\Prob (v\in H_\rho )\geq 1-\tilde{\varepsilon}$ and $\Prob (G\in H_\rho )\geq (1-\tilde{\varepsilon})^n$, where $\tilde{\varepsilon}$ can be determined by $\rho$ and the model's parameters, as per Equation~(\ref{epsilon_bar}). 
	
	For a graph $G$ in the SBM, i.e. $G\in\mathcal{G}(n,k,p,q)$, we can expect to find a $\rho$ such that the communities of $G$ present $\rho$-happy colour-classes for $G$. This is not always possible in the general case. For instance, the star graph $K_{1,m}$ admits no $\rho$-happy colouring for $0<\rho\leq 1$ with more than one colour. However, for a graph in the SBM, we can find a $\xi$ such that for $0<\rho\leq \xi$, the graph's communities induce a colouring that makes almost all vertices $\rho$-happy. Of course, to have a high probability for happiness of an arbitrary vertex $v$ in the colouring induced by communities of $G$, we must have
	\[
	\mathbb{E}\left[ \rho\D (v) -\D_{in} (v)\right] =
	\frac{n}{k}\left(\rho (p + (k-1)q) - p\right)<0,
	\]
	because otherwise, more than half of all vertices are 
	$\rho$-unhappy. Therefore,
	\[\rho\left(p+(k-1)q\right) < p.\]
	Consequently, we  have
	\begin{equation}\label{eq:rho-xi}
		\xi\leq\frac{p}{p+(k-1)q}.
	\end{equation}
	The following theorem gives another upper bound for $\xi$. 
	
	\begin{theorem}
		Given $n$, $k$, $p$, and $q$, there exist a $0\leq \xi\leq 1$ such that, for $\rho\leq\xi$, with high probability the communities of $G\in\mathcal{G}(n,k,p,q)$ induce a $\rho$-happy colouring.
	\end{theorem}
	\begin{proof}
		We need a $0<\rho\leq 1$ so that Equation~\ref{th_eq} holds. Then, by Theorem~\ref{th:SBM-happy}, the probability of the existence of a $\rho$-happy colouring of $G$ is at least $1-n\varepsilon\leq (1-\varepsilon)^n$, which will be close to 1 if $\varepsilon$ is small enough. Hence, we choose a small $\varepsilon>0$ for example $\varepsilon\leq\frac{1}{n^2}$. 
		
		We can find a bound on $\rho$ based on Equation~\ref{th_eq} as follows. Since we want  $$p\cdot \left(e^{\rho}-e\right)+(k-1)q\cdot\left(e^{\rho}-1\right)<\frac{k}{n}\mathrm{ln}(\varepsilon),$$ we can separate terms that contain $e^\rho$ from the left-hand side of the above inequality to have
		$$e^{\rho}\cdot \left(p+(k-1)q\right)<\frac{k}{n}\mathrm{ln}(\varepsilon)+p e +(k-1)q.$$ 
		Therefore, $$e^{\rho} <\frac{\frac{k}{n}\mathrm{ln}(\varepsilon)+p e +(k-1)q}{p+(k-1)q},$$ which means that
		\begin{equation}\label{eq:rho}
			\rho <\mathrm{ln}\left(\frac{\frac{k}{n}\mathrm{ln}(\varepsilon)+p e +(k-1)q}{p+(k-1)q}\right).  
		\end{equation}
		
		Using Equations~(\ref{eq:rho-xi}) and~(\ref{eq:rho}), we now get
		\begin{equation}\label{xi}
			\xi=\max\left\{\min\left\{\mathrm{ln}\left(\frac{\frac{k}{n}\mathrm{ln}(\varepsilon)+p e +(k-1)q}{p+(k-1)q}\right),\; \frac{p}{p+(k-1)q}\right\}, \; 0\right\},
		\end{equation}
		which has the desired property.
	\end{proof}
	
	Therefore, $\xi$ is a threshold so that communities of the graph induce a $\rho$-happy colouring for $0\leq\rho\leq\xi$. Note that we did not exclude the $0$-happy colouring in Definition~\ref{def:rho-happy} since $\xi$ can sometimes be 0 and therefore $0$-happy colouring must be defined. To increase the probability that communities of the graph induce a $\rho$-happy colouring, we may consider $\rho\leq\frac{\xi}{2}$.

	The threshold $\xi$ (Equation~(\ref{xi})) has shown to be useful in our experiments. We investigated whether $\xi$ is a meaningful threshold for communities to induce $\rho$-happy colour classes, and for this purpose, we examined 100,000 graphs in the SBM with 500, 1,000, 2,000, 3,000 and 5,000 vertices with randomly chosen $k\in\{2,3,\ldots, 20\}$, $0<p\leq 1$, $0<q\leq\frac{p}{2}$ and $0<\rho\leq 1$, When $\rho\leq\frac{\xi}{2}$, the number of $\rho$-happy vertices --- in the colouring induced by the graph's communities --- is always high, evidence of this is provided in Figure~\ref{fig:xi-2d}. Moreover, as it can be seen in Figure~\ref{fig:xi-3d}, the minimum number of $\rho$-happy vertices in the same set of graphs are almost always very high when $\xi >\rho$ (yellow dots) while low numbers of $\rho$-happy vertices appear only when $\xi < \rho$ (purple dots). This tells us that $\xi$ presents a useful threshold so that we can expect to have a complete $\rho$-happy colouring when $\rho < \xi$. %For further details of our experiments see Section~\ref{Sec:exprimental}. 

	\begin{figure}
		\centering
		\includegraphics[scale=0.6]{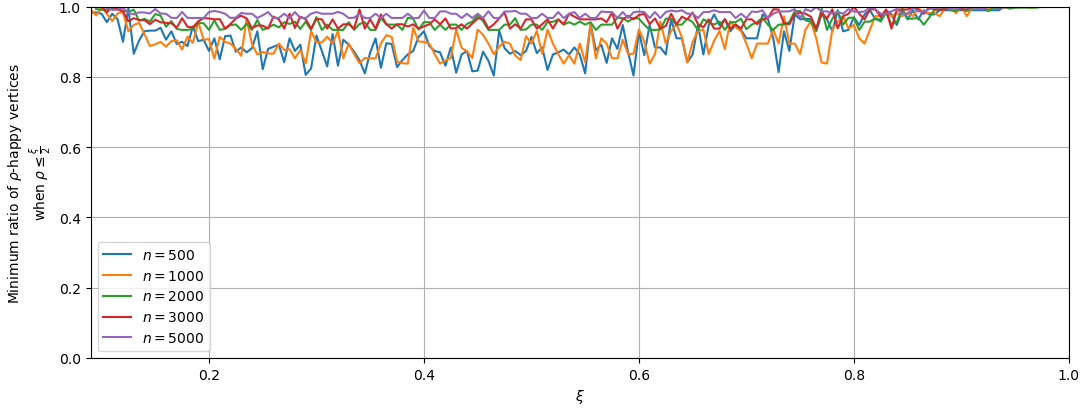}
		\caption{This graph shows the minimum ratio of $\rho$-happy vertices when $\rho\leq\frac{\xi}{2}$ in the induced colouring by graphs communities. Tests were performed over 20,000 randomly generated graphs with $n\in\{$500, 1,000, 2,000, 3,000, 5,000$\}$ (100,000 graphs in total) and parameters randomly selected from $2\leq k\leq 20$, $0<p\leq 1$, and $0<q\leq \frac{p}{2}$.}
		\label{fig:xi-2d}
	\end{figure}
	
	\begin{figure}

    \centering
    \includegraphics[scale=0.9]{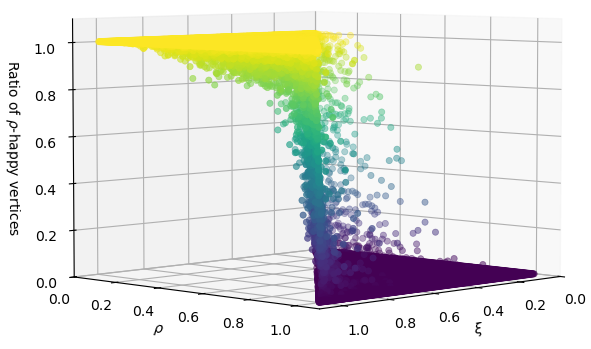}
    \caption{The minimum ratio of $\rho$-happy vertices of the induced colouring by communities of 100,000 randomly generated precoloured graphs with $n\in\{500, 1,000, 2,000, 3,000, 5,000\}$ (20,000 graphs for each of these numbers) while other parameters are randomly chosen from $2\leq k\leq 20$, $0<p\leq 1$, and $0<q\leq \frac{p}{2}$. It can be seen when $\rho<\xi$, the probability of having a large number of $\rho$-happy vertices is high (yellow areas).}
		\label{fig:xi-3d}
	\end{figure}

	Moreover, when we know that the probability of a vertex being $\rho$-happy by the community-induced colouring is at least $1-\varepsilon$, then the minimum probability of such colouring being $\rho$-happy depends not only on $n$ but also on $\varepsilon$. If $n\rightarrow\infty$, then this probability goes to zero when $\epsilon$ remains constant. If $\varepsilon$ decreases harmonically, the probability goes to a positive number between 0 and 1. For example, if $\varepsilon= \frac{1}{n}$, then $(1-\frac{1}{n})^n\rightarrow e^{-1}\simeq 0.36787\textcolor{red}{\ldots}$. Meanwhile, since $n\rightarrow\infty$ results in $(1-\frac{1}{n^t})^n\rightarrow 1$ for $t>1$, the probability approaches 1 when $n$ becomes sufficiently large when $\varepsilon$ drops sharply, for example when $\varepsilon=n^{-2}$. 
	
	To simplify the notations in Equation~(\ref{epsilon}), define
	\begin{equation}
		\varphi=\varphi(k,p,q,\rho)=\frac{1}{k}\cdot \left(p\cdot \left(e^{\rho}-e\right)+q(k-1)\left(e^{\rho}-1\right)\right).
	\end{equation}
	Hence, $$\tilde{\varepsilon}=e^{n\varphi}.$$ It is evident that the lower the value of $n\cdot\varphi$, the higher the following probability 
	\begin{equation}\label{eq:prob_tilde}
		\Prob(G\in H_\rho)\geq (1-\tilde{\varepsilon})^n.
	\end{equation}
	This probability has a direct relationship with $p$ and inverse relationships with $k$, $q$ and $\rho$. Its relation with $n$ is not as straightforward because sometimes increasing the fraction $\frac{n}{k}$ results in a higher probability of a vertex being $\rho$-happy, but reduces the probability that all the vertices remain $\rho$-happy. Our intention here is to consider the asymptotic interpretation of the relation of the number of vertices with $\rho$-happiness of the induced colouring of a graph in the SBM when other parameters remain constant. This is expressed in the following theorem. 
	
	\begin{theorem}\label{th:infinity}
		Let $0<q<p<1$,  $k\in\mathbb{N}\setminus\{1\}$ be constants, 
		and $$\tilde{\xi}=\min\left\{\mathrm{ln}\left(\frac{p e +(k-1)q}{p+(k-1)q}\right),\frac{p}{p+(k-1)q}\right\}.$$ Then, for $0\leq \rho <\tilde{\xi}$ and $G\in\mathcal{G}(n,k,p,q)$, we have $\Prob(G\in H_\rho)\rightarrow 1$ when $n\rightarrow \infty$. In other words, the probability that the communities of $G$ induce a $\rho$-happy colouring on its vertex set approaches 1 when $n$ becomes large enough. 
	\end{theorem}
	\begin{proof}
		For any $0<\varepsilon<1$ we have $$\lim_{n\rightarrow\infty} \mathrm{ln}\left(\frac{\frac{k}{n}\mathrm{ln}(\varepsilon)+p e +(k-1)q}{p+(k-1)q}\right)=\mathrm{ln}\left(\frac{p e +(k-1)q}{p+(k-1)q}\right).$$
		As a result, $\lim_{n\rightarrow\infty} \xi=\tilde{\xi}$. And for $\rho<\tilde{\xi}$, we not only have \linebreak $\mathbb{E}\left[ \rho\D (v) -\D_{in} (v)\right] <0$, but also 
		\[\begin{split}
			& e^\rho <\frac{p e +(k-1)q}{p+(k-1)q}\\
			\implies & e^\rho (p+(k-1)q) - pe -(k-1)q <0\\
			\implies & \varphi(k,p,q,\rho)<0.
		\end{split}\]
		Therefore, as $n$ grows, $1-\tilde{\varepsilon}$ also increases (see Equation~(\ref{epsilon_bar})). In particular, since $$\lim_{n\rightarrow\infty} n\cdot\varphi(k,p,q,\rho)=-\infty,$$   by Equation~\ref{eq:prob_tilde}, we have
		\begin{equation*}
			\begin{split}
				\lim_{n\rightarrow \infty}\Prob(G\in H_\rho)&\geq \lim_{n\rightarrow \infty} (1-\tilde{\varepsilon})^n\\
				&=\lim_{n\rightarrow \infty}(1-e^{n\varphi})^n\\
				&=\lim_{n\rightarrow \infty} e^{-n e^{n\varphi}}\\
				&=\mathrm{EXP}\left(\lim_{n\rightarrow \infty} -n e^{n\varphi}\right) \\
				&=\mathrm{EXP}\left(-\lim_{n\rightarrow \infty} \frac{n}{e^{-n\varphi}}\right).
			\end{split}
		\end{equation*}
		Therefore, by l'Hospital's rule, we have
		\begin{equation*}
			\lim_{n\rightarrow \infty}\Prob(G\in H_\rho)\geq \mathrm{EXP}\left(-\lim_{n\rightarrow\infty}\frac{1}{-\varphi e^{-n\varphi}}\right)=e^0=1. 
		\end{equation*}
		Consequently, when $0\leq \rho< \tilde{\xi}$, we have $\lim_{n\rightarrow \infty}\Prob(G\in H_\rho)=1$.
		
	\end{proof}
	
	Figure~\ref{fig:n-r-comm} illustrates how increasing the number of vertices $n$ affects the number of $\rho$-happy vertices of the colouring induced by the communities of graphs in $\mathcal{G}(n,20,0.7,0.06)$ (i.e, $k=20$, $p=0.7$ and $q=0.06$). For $\rho<\tilde{\xi}$, where
	\[ \begin{split}
		\tilde{\xi}&=\min\left\{\mathrm{ln}\left(\frac{p e +(k-1)q}{p+(k-1)q}\right),\;\frac{p}{p+(k-1)q}\right\}\\ &=\min\{\mathrm{ln}(1.6547), 0.3804\}=0.3804,
	\end{split}\]
	we must have $\lim_{n\rightarrow \infty}\Prob(G\in H_\rho)=1$. The phase transition of trends is obvious when $\rho\approx 0.38$ (see the almost flat trend in Figure~\ref{fig:n-comm-2d}).

	\begin{figure}[h!]
		\centering
		
		\begin{subfigure}{0.9\textwidth}
			\includegraphics[scale=0.5]{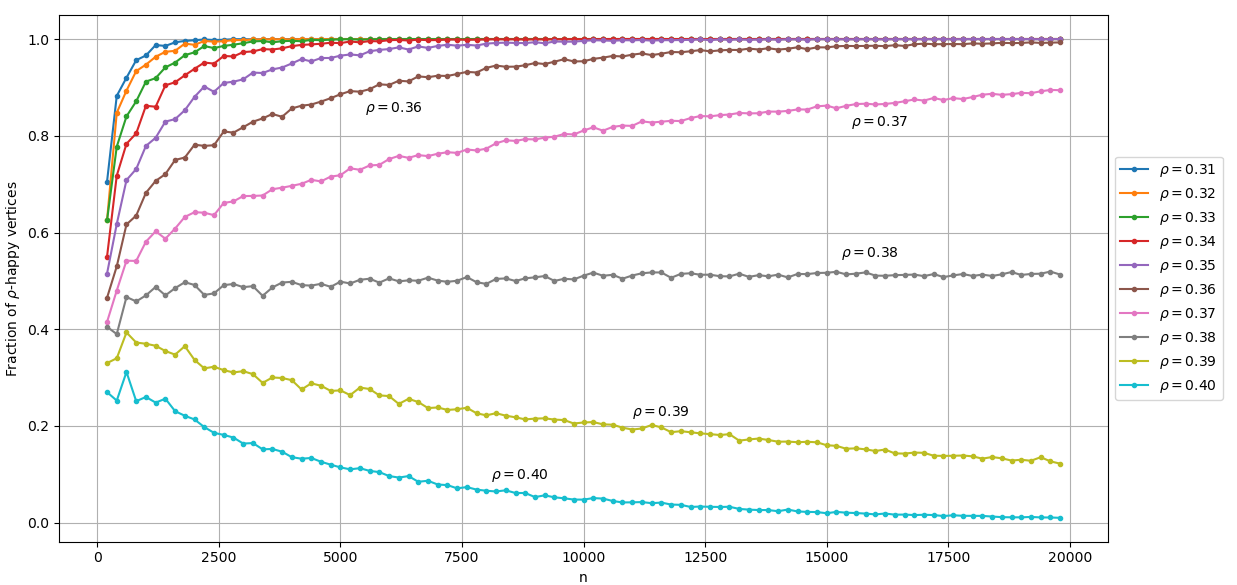}
			\caption{}
			\label{fig:n-comm-2d}
		\end{subfigure}

		\begin{subfigure}{0.75\textwidth}
			\includegraphics[scale=0.8]{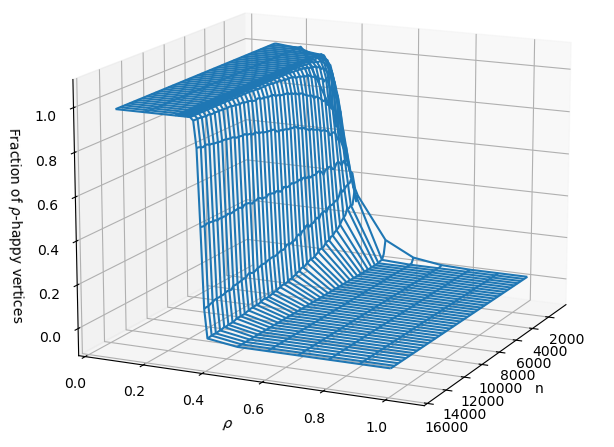}
			\caption{}
			\label{fig:n-r-comm-xi}
		\end{subfigure} 
		\caption{The charts show the effects of increasing the number of vertices on the average number of $\rho$-happy vertices in the colouring induced by communities for graphs in the SBM. Here, $k=20$, $p=0.7$ and $q=0.06$. The chart (a) shows an image of these effects when $200\leq n\leq$ 20,000 and $0.31\leq\rho\leq 0.4$. Part (b) shows these effects when $0.1\leq\rho\leq 1$. When $\rho\leq 0.37$ and $n$ grows, the average number of $\rho$-happy vertices in the colouring induced by communities also grows, while the graphs are entire $\rho$-happily coloured by their communities when $\rho\leq 0.35$ and $n\geq$ 10,000. Interestingly, this number decreases when $\rho>0.38$. When $\rho = 0.38$, the chart is still increasing, but at a very slow rate.}\label{fig:n-r-comm}
	\end{figure}

	\section{Algorithms}\label{sec:alg}
	
	In this section, we propose four heuristic algorithms to tackle the problem of maximising the number of $\rho$-happy vertices and examining the correlation of this problem and theoretical results from the previous section.	In the proposed algorithms, we denote the number of $\rho$-happy vertices in a colouring $c$ of a graph $G$ by $H(G,c,\rho)$.  
	
	\subsection{\sf Greedy-SoftMHV}

Algorithm~\ref{greedy1} is a straightforward update of {\sf GreedyMHV} from~\cite{ZHANG2015117}, the only difference being the change from ``happiness'' to ``$\rho$-happiness.'' It turns a partial colouring into a $k$-colouring, where only one colour is assigned to all the uncoloured vertices, ensuring the chosen colour makes as many vertices $\rho$-happy as possible. The algorithm is polynomial, i.e., $\mathcal{O}(km)$~\cite{Carpentier2023} where $m$ is the number of edges and $k$ is the number of permissible colours. 

The algorithm works as follows. As input, it takes a graph $G$, the parameter that determines the proportion of happiness $\rho$ and a partial colouring $c$. The initialisation takes place in Lines 1--2, including assigning precolours to vertices in $V$, the set of uncoloured vertices $U$ and $\text{\emph{Max}} \leftarrow 0$ which records the maximum number of happy vertices. Between Lines 3--11, each colour is selected and assigned to the uncoloured vertices, and the total number of happy vertices is computed (Line 6). If an improvement is found, it is recorded via the variables $\text{\emph{Max}}$ and $i_{\text{\emph{Max}}}$. The output of the algorithm is a complete colouring $\tilde{c}$.

\def\NoNumber#1{{\def\alglinenumber##1{}\State #1}\addtocounter{ALG@line}{-1}}
\begin{algorithm}
	\caption{({\sf Greedy-SoftMHV}) Approximating the optimal solution of $\rho$-happy colouring~\cite{ZHANG2015117} \label{greedy1}}
	\begin{flushleft} $\;$\\ \hspace*{\algorithmicindent}
		\textbf{Input:} $G$, $\rho$, $c:V'\longrightarrow \{1,\ldots,k\}$ \Comment{$V'\subsetneq V(G)$}\\ 
		\hspace*{\algorithmicindent}  \textbf{Output:} $\tilde{c}:V(G)\longrightarrow \{1,\ldots,k\}$ \Comment{a complete colouring of $G$}
	\end{flushleft}
	\begin{algorithmic}[1]
		
		\State $\forall i\in\{1,\ldots,k\}, V_i\gets\{v:c(v)=i\}$ 
		\State $U\gets V(G)-V'$, $\text{\emph{Max}}\gets 0$ and $i_{\text{\emph{Max}}}\gets 0$ \Comment{$U=$ uncoloured vertices}
		\NoNumber{ }
		
		\For{$i=1,\ldots,k$} \Comment{try all colours}
		
		\State {\bf Append} members of $U$ to $V_i$ 
		\State $\tilde{c}\gets \{V_1,\ldots,V_k\}$
		\State {\bf Calculate} $H(G,\tilde{c},\rho)$ \Comment{the number of $\rho$-happy vertices in $\tilde{c}$}
		\If{$\text{\emph{Max}}<H(G,\tilde{c},\rho)$}
		\State $\text{\emph{Max}}\gets H(G,\tilde{c},\rho)$ and $i_{\text{\emph{Max}}}\gets i$
		%\State $i_{Max}\gets i$
		\EndIf
		\State {\bf Remove} members of $U$ from $V_i$
		\EndFor
		\State  {\bf Append} members of $U$ to $V_{i_{\text{\emph{Max}}}}$ 
		
		\State {\bf Return} $\tilde{c}=\{V_1,\ldots,V_k\}$
	\end{algorithmic}
\end{algorithm}

\subsection{\sf Neighbour Greedy Colouring (NGC)} 

At the cost of increased computational time, we propose \hyperref[greedy2]{\sf NGC}, in which, unlike the previous approach, all uncoloured vertices do not always receive the same colour. Its initialization and colour assignment are exactly like  \hyperref[greedy1]{\sf Greedy-SoftMHV}. The only difference is that at each step, it only colours neighbours of already coloured vertices.

Like the \hyperref[greedy1]{\sf Greedy-SoftMHV}, the \hyperref[greedy2]{\sf NGC} takes a graph $G$, $0\leq \rho\leq 1$, and a partial colouring $c$.  
Then, it initialises variables $U$ as the set of uncoloured vertices, $\text{\emph{Max}}$ as the maximum number of $\rho$-happy vertices until now, and $i_{\text{\emph{Max}}}$ as the colour that makes this maximum.
In Lines 3--12, it colours all the uncoloured vertices to see for which colour, say $i_{\text{\emph{Max}}}$, the number of $\rho$-happy vertices is the largest. Afterwards in Lines 13--14, only the uncoloured neighbours of already coloured vertices by $i_{\text{\emph{Max}}}$ receive the colour $i_{\text{\emph{Max}}}$. The procedure runs until all uncoloured vertices receive a colour, and then reports the complete colouring $\tilde{c}$. 

\begin{algorithm}
	\caption{({\sf Neighbour Greedy Colouring}) Approximating the optimal solution of $\rho$-happy colouring}\label{greedy2}
	\begin{flushleft} $\;$\\ \hspace*{\algorithmicindent}
		\textbf{Input:} $G$, $\rho$, $c:V'\longrightarrow \{1,\ldots,k\}$ \Comment{$V'\subsetneq V(G)$} \\
		\hspace*{\algorithmicindent} \textbf{Output:} $\tilde{c}:V(G)\longrightarrow \{1,\ldots,k\}$ \Comment{a complete colouring of $G$}
	\end{flushleft}
	\begin{algorithmic}[1]
		
		\State $\forall i\in\{1,\ldots,k\}, V_i\gets\{v:c(v)=i\}$ 
		\State $U\gets V(G)-V'$, $\text{\emph{Temp}}\gets \emptyset$, $\text{\emph{Max}}\gets 0$  and $i_{\text{\emph{Max}}}\gets 0$\Comment{$U=$ uncoloured vertices}
		\NoNumber{ }
		\While{$U\neq \emptyset$}
		\For{$i=1,\ldots,k$} \Comment{try all colours}
		
		\State {\bf Append} members of $U$ to $V_i$ 
		\State $\tilde{c}\gets \{V_1,\ldots,V_k\}$
		\State {\bf Calculate} $H(G,\tilde{c},\rho)$ \Comment{the number of $\rho$-happy vertices in $\tilde{c}$}
		\If{$\text{\emph{Max}}<H(G,\tilde{c},\rho)$}
		\State $\text{\emph{Max}}\gets H(G,\tilde{c},\rho)$ and $i_{\text{\emph{Max}}}\gets i$
		\EndIf
		\State {\bf Remove} members of $U$ from $V_i$
		\EndFor
		\State $\text{\emph{Temp}}\gets U\cap N(V_{i_{\text{\emph{Max}}}})$
		\State {\bf Remove} members of \emph{Temp} from $U$ and {\bf Append} members of \emph{Temp} to $V_{i_{\mathrm{Max}}}$
		\EndWhile
		\State {\bf Return} $\tilde{c}=\{V_1,\ldots,V_k\}$
	\end{algorithmic}
\end{algorithm}

The time complexity of the \hyperref[greedy2]{\sf NGC} is $d=\mathrm{diam}(G)$ times the time complexity of \hyperref[greedy1]{\sf Greedy-SoftMHV}, or $\mathcal{O}(dkm)$. This is because in the worst case, it has to repeat $d$ times the processes 
of \hyperref[greedy1]{\sf Greedy-SoftMHV} in Lines 4--12 to colour all the vertices. 

Moreover, the solution quality of the \hyperref[greedy2]{\sf NGC} is at least as good as that of \hyperref[greedy1]{\sf Greedy-SoftMHV}. This is because if the \hyperref[greedy2]{\sf NGC} ends with assigning only one colour to all non-precoloured vertices, then it gives the exact output as \hyperref[greedy1]{\sf Greedy-SoftMHV}. If \hyperref[greedy2]{\sf NGC} ends with assigning more than two colours to non-precoloured vertices,  the number of $\rho$-happy vertices is higher than that of the output of the \hyperref[greedy1]{\sf Greedy-SoftMHV}.  

\subsection{{\sf Local Maximal Colouring}} 

Our third proposed heuristic algorithm, \hyperref[greedy3]{\sf Local Maximal Colouring (LMC)}, finds $\rho$-happy colour classes by examining local neighbourhoods of the uncoloured vertices, identifying the most frequently appearing colour, and assigning this colour to all uncoloured vertices in this neighbourhood.

The \hyperref[greedy3]{\sf LMC}'s inputs are a graph $G$ and a partial colouring $c$. Then, it employs two variables in Lines 2--3, namely $U$ and $C$, as the set of uncoloured vertices and the set of coloured ones, respectively. Afterwards, an uncoloured vertex $v$, adjacent to at least one coloured vertex, is chosen. Then in Lines 6--7, $q$ is the most frequent colour in $N(v)$, and $v$ receives the colour $q$. The vertex $v$ is now coloured, which is added to the set of coloured vertices $C$ and removed from the uncoloured vertices. This procedure continues until no uncoloured vertex remains. 

\begin{algorithm}
\caption{({\sf Local Maximal Colouring}) Approximating the optimal solution of $\rho$-happy colouring.}\label{greedy3}
\begin{flushleft} $\;$\\ \hspace*{\algorithmicindent}
	\textbf{Input:} $G$, $c:V'\longrightarrow \{1,\ldots,k\}$ \Comment{$V'\subsetneq V(G)$}\\
	\hspace*{\algorithmicindent} \textbf{Output:} $\tilde{c}:V(G)\longrightarrow \{1,\ldots,k\}$ \Comment{a complete colouring of $G$}
\end{flushleft}
\begin{algorithmic}[1]
	
	\State $\forall i\in\{1,\ldots,k\}, V_i\gets\{v:c(v)=i\}$
	\State $U\gets  V(G)-V'$\Comment{$U=$ uncoloured vertices}
	\State $C\gets V_1 \cup \ldots \cup V_k$ %\AN{C was used for something else before.}
	\Comment{$C=$ coloured vertices} 
	\NoNumber{ }
	\While{$U\neq \emptyset$}
	\State {\bf Choose} $v\in U\cap N(C)$
	
	\State $q\gets $ the colour which appears the most in $N(v)$
	\State {\bf Append} $v$ to $V_{q}$, {\bf Append} $v$ to $C$ and {\bf Remove} $v$ from $U$
	\EndWhile
	\State {\bf Return} $\tilde{c}=\{V_1,\ldots,V_k\}$
\end{algorithmic}
\end{algorithm}

The \hyperref[greedy3]{\sf LMC} has a time complexity of $\mathcal{O}(m)$, owing to the main loop (Lines~4--10), where the colours of neighbours of $v$ are examined. This happens at most $2m$ times, i.e., twice for each edge. Experiments (Section~\ref{Sec:exprimental}) show that \hyperref[greedy3]{\sf LMC} can detect community structures close to the graphs' true community structures. Moreover, its time complexity, near-linear, is the lowest among the algorithms investigated. However, since it does not rely on $\rho$, its output can occasionally have fewer $\rho$-happy vertices than other algorithms.

\subsection{\sf Growth-SoftMHV}

Algorithm~\ref{growth} has been adapted from the {\sf Growth-SoftMHV} algorithm of~\cite{ZHANG2015117}  
to solve Soft happy colouring. The main process of this algorithm is determining the vertex classifications introduced in the following definition.
	
	\begin{definition}\label{Def:Growth} \textnormal{\cite[altered from Definitions 8 and 9]{ZHANG2015117}}
		Let $G$ be a graph whose vertices are partially coloured by $c:V'\longrightarrow\{1,\ldots,k\}$, $V'\subsetneq V(G)$, $V_1,\ldots,V_k$ are colour classes, $U$ is the set of uncoloured vertices, $v\in V(G)$ and $0<\rho \leq 1$.
		\begin{itemize}
			\item[1.] $v$ is an H-vertex if it is coloured and $\rho$-happy.
			\item[2.] $v$ is a U-vertex if
			\begin{itemize}
				\item[2.a.] $v$ is coloured, and
				\item[2.b.] $v$ is destined to be $\rho$-unhappy, (i.e., $| N(v) \cap V_{c(v)}| +| N(v) \cap U| < \rho \cdot \deg (v)$)
			\end{itemize}
			\item[3.] $v$ is a P-vertex if
			\begin{itemize}
				\item[3.a.]  $v$ is coloured,
				\item[3.b.] $v$ has not been $\rho$-happy (i.e., $|N(v)\cap  V_{c(v)}| < \rho\cdot \deg (v)$), and
				\item[3.c.] $v$ can become an $H$-vertex (i.e., $| N(v) \cap V_{c(v)}| + | N(v)\cap U| \geq \rho \cdot \deg (v)$)
			\end{itemize}
			\item[4.] $v$ is an L-vertex if it has not been coloured.
			\begin{itemize}
				\item[4.1.] $v$ is an $\mathrm{L}_\mathrm{p}$-vertex if it is adjacent to a P-vertex,
				\item[4.2.] $v$ is an $\mathrm{L}_\mathrm{h}$-vertex if
				\begin{itemize}
					\item[4.2.a.] $v$ is not adjacent to any P-vertex,
					\item[4.2.b.] $v$ is adjacent to an H-vertex or a U-vertex, and
					\item[4.2.c.] $v$ can become $\rho$-happy, that is, $$|N (v)\cap U| + \max\{| N (v)\cap V_i|: 1 \leq i \leq k\} \geq \rho \cdot \deg (v).$$
				\end{itemize}
				\item[4.3.] $v$ is an $\mathrm{L}_\mathrm{u}$-vertex if
				\begin{itemize}
					\item[4.3.a.] $v$ is not adjacent to any P-vertex,
					\item[4.3.b.] $v$ is adjacent to an H-vertex or a U-vertex, and
					\item[4.3.c.] $v$ is destined to be $\rho$-unhappy, that is, $$|N(v)\cap U| + \max\{| N(v)\cap V_i|: 1 \leq i \leq k\} < \rho\cdot \deg (v).$$
				\end{itemize}
				\item[4.4.] $v$ is an $\mathrm{L}_\mathrm{f}$-vertex if it is not adjacent to a coloured vertex.
			\end{itemize}
		\end{itemize}
	\end{definition}

The inputs of the \hyperref[growth]{\sf Growth-SoftMHV} algorithm are a graph $G$, $\rho$, and a partial $k$-colouring $c$. It first generates the colour classes (Line 1) and identifies all $P$, $L_h$, and $L_u$-vertices in Line 2, where it assigns them to sets $L=P\cup L_h \cup L_u$. Following this in Lines 4--8, it checks if there is a $P$-vertex $v$, it then chooses just enough neighbours of $v$ and colours them by the colour of $v$ so that $v$ becomes $\rho$-happy. 
Then it recalculates the vertex classes $P$, $L_h$, $L_u$, and $L$. The above steps may generate other $P$-vertices, hence, the procedure repeats until no $P$-vertex remains.

After running out of  $P$-vertices, the algorithm checks in Lines 9--15 if there is an $L_h$ vertex $v$ and chooses a colour $i$ that appears the most among the neighbours of $v$. Then it colours $v$ and some of its neighbours by the colour $i$ so that $v$ becomes a $\rho$-happy vertex. Next, because these colour assignments might generate new $P$ or $L_h$-vertices, the vertex classification of $P$, $L_h$, $L_u$ and $L$ is updated. The algorithm repeats Lines 4--8 and 9--15 
until no $P$ or $L_h$-vertex is left. 

\begin{algorithm}
\caption{({\sf Growth-SoftMHV}) Approximating the optimal solution of soft happy colouring}\label{growth}
\begin{flushleft} $\;$\\ \hspace*{\algorithmicindent}
\textbf{Input:} $G$, $\rho$, $c:V'\longrightarrow \{1,\ldots,k\}$ \Comment{$V'\subsetneq V(G)$} \\
\hspace*{\algorithmicindent} \textbf{Output:} $\tilde{c}:V(G)\longrightarrow \{1,\ldots,k\}$ \Comment{a complete colouring of $G$}
\end{flushleft}
\begin{algorithmic}[1]

\State $\forall i\in\{1,\ldots,k\}, V_i\gets\{v:c(v)=i\}$ %\Comment{start of variable initialization}
\State {\bf Calculate} $P, L_h, L_u$ and $L$ \Comment{Definition \ref{Def:Growth}}
\NoNumber{ }
\While{$L\neq \emptyset$}
\While{$P\neq \emptyset$}
\State {\bf Choose} $v\in P$
\State {\bf Choose} $\left\lceil \rho\cdot \deg (v)\right\rceil - |N(v)\cap V_{c(v)}|$ from $N(v)\cap L$ and {\bf Append} them to $V_{c(v)}$
\State {\bf ReCalculate} $P, L_h, L_u$ and $L$ \Comment{Definition \ref{Def:Growth}}
\EndWhile
\While{($P=\emptyset$ and $L_h\neq \emptyset$)}
\State {\bf Choose} $v\in L_h$
\State {\bf Choose} $i\in\{1,\ldots,k\}$ such that $\forall t, |V_i\cap N(v)|\geq |V_t \cap N(v)|$
\State {\bf Append} $v$ to $V_i$
\State {\bf Choose} $\left\lceil \rho\cdot \deg (v)\right\rceil - |N(v)\cap V_i|$ from $N(v)\cap L$ and {\bf Append} them to $V_i$
\State {\bf ReCalculate} $P, L_h, L_u$ and $L$ \Comment{Definition \ref{Def:Growth}}
\EndWhile
\While{($P=\emptyset$ and $L_h=\emptyset$ and $L_u\neq \emptyset$)}
\State {\bf Choose} $v\in L_u$
\State {\bf Choose} $i\in\{1,\ldots,k\}$ such that $\forall t, |V_i\cap N(v)|\geq |V_t \cap N(v)|$
\State {\bf Append} $v$ to $V_i$
\State {\bf Choose} $\left\lceil \rho\cdot \deg (v)\right\rceil - |N(v)\cap V_i|$ from $N(v)\cap L$ and {\bf Append} them to $V_i$
\State {\bf ReCalculate} $P,L_h, L_u$ and $L$ \Comment{Definition \ref{Def:Growth}}
\EndWhile
\EndWhile
\State {\bf Return} $\tilde{c}=\{V_1,\ldots,V_k\}$
\end{algorithmic}
\end{algorithm}

If a complete colouring has not yet been obtained, at least one $L_u$-vertex must exist because the graph is assumed to be connected and $L_f$-vertices have no coloured neighbours. 
The algorithm in Lines 16--22 chooses an $L_u$-vertex $v$ and a colour $i$ that appears the most among neighbours of $v$, then the colour of $v$ and some of its uncoloured vertices turns into the colour $i$. After recalculating the vertex classifications $P$, $L_h$, $L_u$, and $L$, the algorithm repeats Lines 4--22 
until no $P$, $L_h$, or $L_u$ vertex is left. 
Now, Every vertex must have a colour, so the algorithm reports the complete colouring. The time complexity of \hyperref[growth]{\sf Growth-SoftMHV} has been calculated by Carpentier et al.~\cite{Carpentier2023} as $\mathcal{O}(mn)$.

\section{Experimental evaluation}\label{Sec:exprimental}

We conducted experiments on various graphs based on the SBM to investigate the performance of the proposed algorithms and validate our theoretical results. We developed a problem generator, from which we generated %19,000 partially coloured graphs with 1000 vertices
our sample graphs.\footnote{The graphs are generated using \texttt{Python} under 
the \texttt{stochastic\textunderscore block\textunderscore model} from the package \texttt{NetworkX}} The problem generator, the source code for the algorithms and the graphs tested (stored in DIMACS format) are available online.\footnote{at \href{https://github.com/mhshekarriz/HappyColouring_SBM}{https://github.com/mhshekarriz/HappyColouring\_SBM}} 
Here, we used a computer running $12\times 3.60$~GHz Intel\texttrademark Xeon\texttrademark CPUs, 32 GB of RAM, and 512 GB memory 
to generate the problem instances and run all algorithms.

To test the algorithms, we generated partially coloured sample graphs considering the following parameters: number of vertices $n$, number of partitions (communities or colour classes) $k$, the edge probability inside communities $p$, the probability of inter-community edges $q$, the number of precoloured vertices per community $pcc$, a random seed for generating the graphs and the proportion of happiness $\rho$. Because increasing $n$ impacts the running time of the algorithms, we split the test into two parts. First, we check the algorithms on 19,000 partially coloured graphs with 1,000 vertices on several combinations of other parameters. This enables us to see the effects of changing other parameters more clearly. Afterwards, 28,000 partially coloured graphs on $200\leq n <$ 3,000 vertices are tested (10 graphs for each $n$), while other parameters are chosen at random. In both tests, we make sure that no two vertices in a community are precoloured with different colours. This is to prevent the induced precolouring from contradicting the community structure of the graphs. Otherwise, we cannot expect that the induced colouring by communities makes the entire vertex set $\rho$-happy. The settings for the two tests are as follows.

\begin{itemize}
    \item[{\bf Test 1.}] Graphs for this test were generated with the settings $n=$ 1,000, $k=2,3,\ldots, 20$, $p=0.1,0.2,\ldots,0.9$, $q=0.01,0.11,0.21, \ldots, \leq \frac{p}{2}$. For each combination of these values, four instances were generated. The number of precoloured vertices in each community varied from 1 to 10. Each of the 19,000 graphs is then tested for $\rho = 0.1, 0.2, \ldots, 1$. The time limit is set to 40 seconds. The results of running each of the algorithms 190,000 times in total are summarised in Section~\ref{sec:1000}.

    \item[{\bf Test 2.}] For this test graphs are randomly generated with $200\leq n<$ 3,000 vertices. For each $n$, 10 instances are generated, which makes it a set of 28,000 randomly generated graphs. For each graph, other parameters are randomly chosen over the same interval of the previous test i.e. $k\in\{2,3,\ldots,20\}$, $p\in (0,1]$, $q\in (0, \frac{p}{2}]$ and $\rho\in (0,1]$. The time limit for this test is 120 seconds. The results from this test are summarised in Section~\ref{sec:scale}.
\end{itemize}

\subsection{Test 1: graphs on 1,000 vertices}\label{sec:1000}

Perhaps one of the most important test results for us, as anticipated by Theorem~\ref{th:SBM-happy}, is verifying that if $k$ and $\rho$ are small enough, then communities of networks in the SBM induce a $\rho$-happy colouring. The results of the first test fully supports this, see Figure~\ref{fig:k-r-happy}. On the other hand, when $k$ or $\rho$ increases, the induced colouring by graph communities cannot make all the vertices $\rho$-happy. Especially when $\rho=1$, the induced colouring makes almost no vertex happy. 

\begin{figure}
% \captionsetup{size=small}
% 		\begin{subfigure}{0.5\textwidth}
% 			\includegraphics[scale=0.8]{r-k-comm.png}
% 			\caption{}
% 			\label{fig:k-r-happy-n}
% 		\end{subfigure}
%   \hfill
%   		\begin{subfigure}{0.5\textwidth}
\centering
			\includegraphics[scale=0.6]{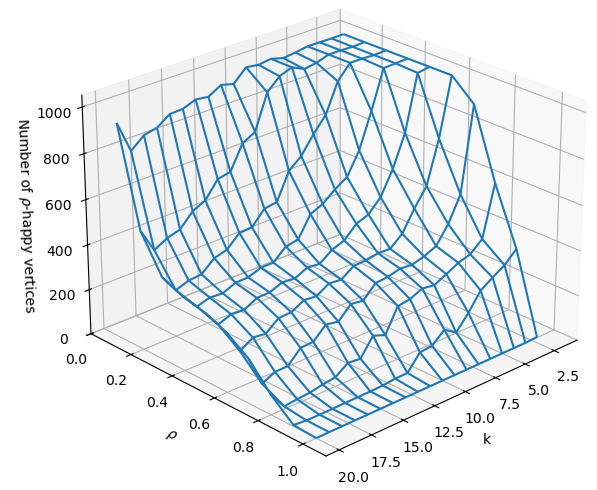} 
		% 	\caption{}
		% 	\label{fig:k-r-happy-1000}
		% \end{subfigure}

\caption{The average number of $\rho$-happy vertices of the induced colouring by communities with respect to $\rho$ and the number of colours.}    \label{fig:k-r-happy}
\end{figure} 

Figure~\ref{fig:q-p-com} illustrates how the average number of $\rho$-happy vertices in the colouring induced by the tested graphs' communities can be affected by the fraction $\frac{p}{q}$. As expected, when $\frac{p}{q}$ is large enough, the communities of the graphs are more straightforward to detect. Consequently, the average number of $\rho$-happy vertices in the colouring induced by the communities is directly related to $\frac{p}{q}$.

\begin{figure}
\centering
\includegraphics[scale=0.6]{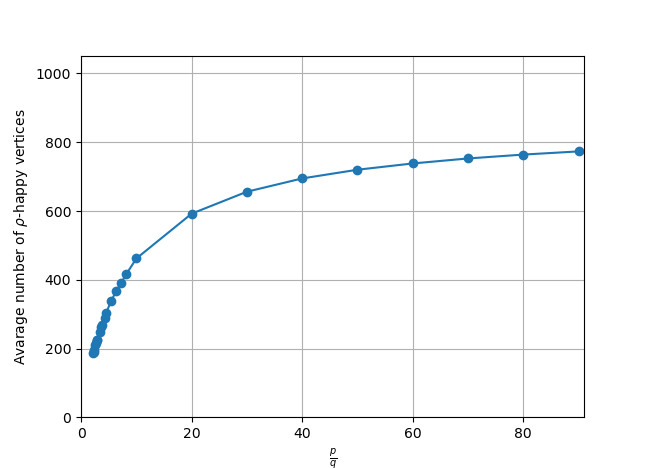}
\caption{The average number of $\rho$-happy vertices of colourings induced by communities with respect to $\frac{p}{q}$. Here, the average is taken for all tested $\rho$s, including the case $\rho=1$ for which we almost always have a lot of unhappy vertices.  }
\label{fig:q-p-com}
\end{figure}

Figure~\ref{fig:3d} demonstrates how changes in parameters affect the average number of $\rho$-happy vertices in the colouring induced by the graphs' communities. Figure~\ref{fig:p-k-comm} indicates that this number increases when $p$ increases. Figure~\ref{fig:q-k-comm} emphasises how increasing $q$ makes the average number of $\rho$-happy vertices drop. These two figures, and Figures~\ref{fig:k-r-happy} and~\ref{fig:pcc-k-comm} as well, demonstrate the negative effect of increasing the number of communities, $k$, on the average number of $\rho$-happy vertices induced by these communities. Figures~\ref{fig:r-p-comm} and~\ref{fig:r-q-comm}, and Figures~\ref{fig:k-r-happy} and~\ref{fig:r-pcc-comm} as well, illustrate a similar trend for $\rho$. The insignificant relation of the average number of $\rho$-happy vertices induced by graphs' communities and the number of precoloured vertices per community can be seen in Figures~\ref{fig:r-pcc-comm} and~\ref{fig:pcc-k-comm}. All these trends agree with the theoretical results in Section~\ref{sec:soft_SBM}. 

\begin{figure}
\captionsetup{size=small}
\begin{subfigure}{0.4\textwidth}
\includegraphics[scale=0.6]{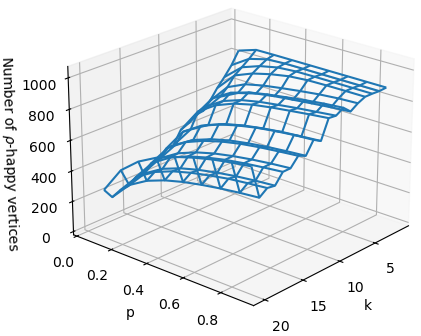}
\caption{}\label{fig:p-k-comm}
\end{subfigure} 
\hfill
\begin{subfigure}{0.4\textwidth}
\includegraphics[scale=0.6]{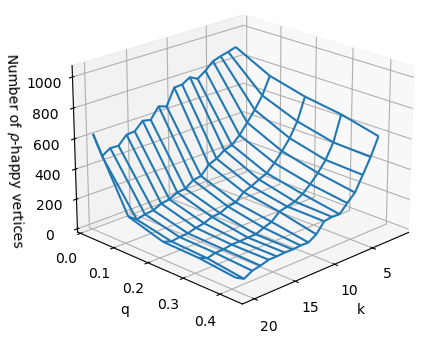}
\caption{}\label{fig:q-k-comm}
\end{subfigure} 

\begin{subfigure}{0.4\textwidth}  
\includegraphics[scale=0.6]{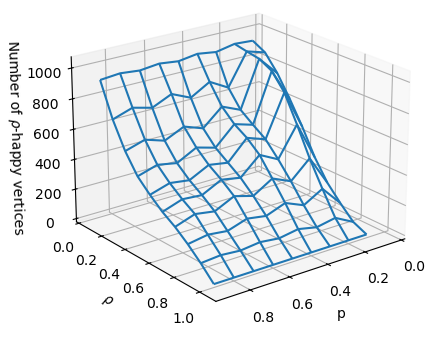}
\caption{}\label{fig:r-p-comm}
\end{subfigure} 
\hfill
\begin{subfigure}{0.4\textwidth}  
\includegraphics[scale=0.6]{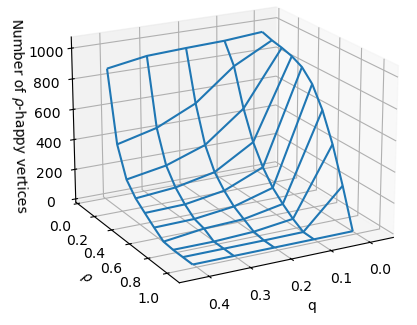}
\caption{}\label{fig:r-q-comm}
\end{subfigure} 

\begin{subfigure}{0.4\textwidth}      
\includegraphics[scale=0.6]{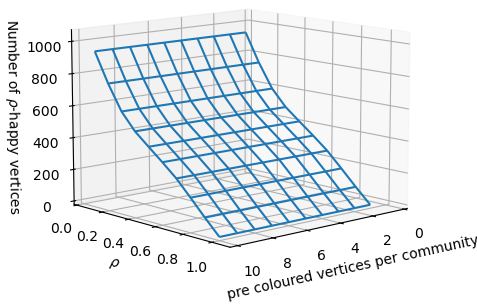} 
\caption{}\label{fig:r-pcc-comm}
\end{subfigure} 
\hfill
\begin{subfigure}{0.4\textwidth}   
\includegraphics[scale=0.6]{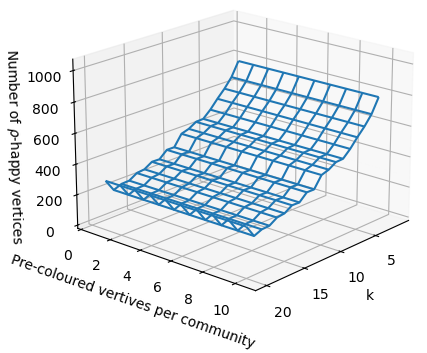}
\caption{}\label{fig:pcc-k-comm}
\end{subfigure} 

\caption{These 3-dimensional charts show the average number of $\rho$-happy vertices when communities induce colour classes compared for (a) the parameters $k$ and $p$, (b) for $k$ and $q$, (c) for $\rho$ and $p$, (d) for $\rho$ and $q$, (e) for $\rho$ and the number of precoloured vertices per community and (f) for $k$ and the number of precoloured vertices per community.}
\label{fig:3d}
\end{figure}

Figures~\ref{fig:time}, \ref{fig:happy} and \ref{fig:comm}, respectively, show the average CPU runtime, the average quality of performance and the average quality of community detection of Algorithms~\ref{greedy1} to~\ref{growth}. The measure of performance is the number of $\rho$-happy vertices in their outputs, and their measure of community detection is the fraction of vertices whose colours in the output graph agree with their communities in the generated graph.  Detecting communities is not explicitly an objective of these algorithms, nonetheless, all algorithms (especially LMC) can effectively detect communities.

Increasing the number of precoloured vertices per community does not significantly affect the run times, %and especially   \hyperref[greedy1]{\sf Greedy-SoftMHV} and \hyperref[greedy3]{\sf LMC} where there is no difference (
see Figure~\ref{fig:pcc-time}. It negatively impacts the number of $\rho$-happy vertices detected by the algorithms (see Figure~\ref{fig:pcc-happy}), while detecting communities improves (see Figure~\ref{fig:pcc-comm}).

\begin{figure}%[!ht]
\captionsetup{size=small}
\begin{subfigure}{0.5\textwidth}
\includegraphics[scale=0.5]{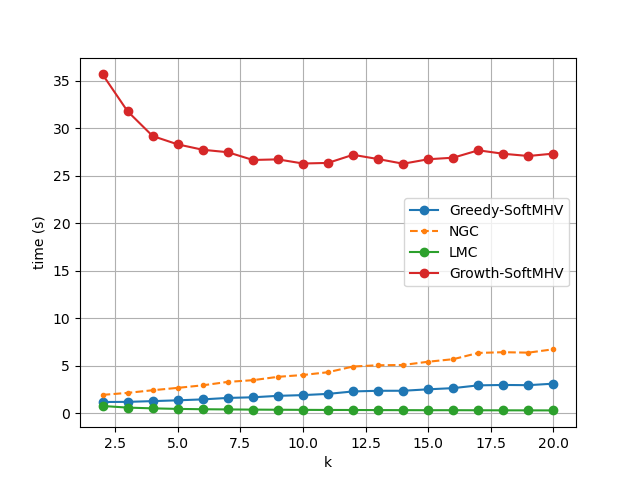}     
\caption{}\label{fig:k-time}
\end{subfigure} 
\hspace{2mm}
\begin{subfigure}{0.5\textwidth}          
\includegraphics[scale=0.5]{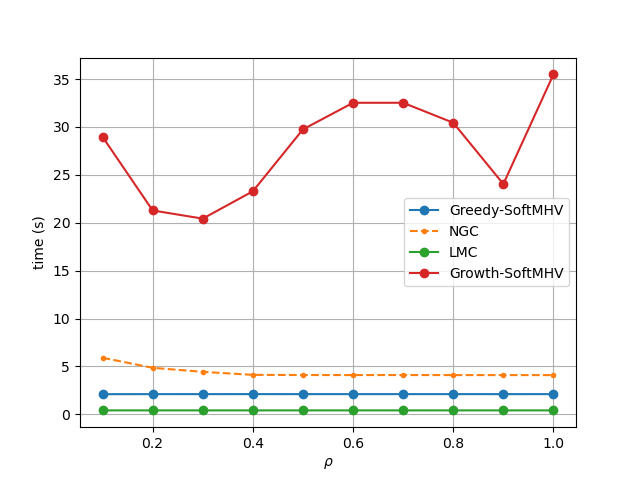} 
\caption{}\label{fig:r-time}
\end{subfigure} 

\begin{subfigure}{0.5\textwidth}  
\includegraphics[scale=0.5]{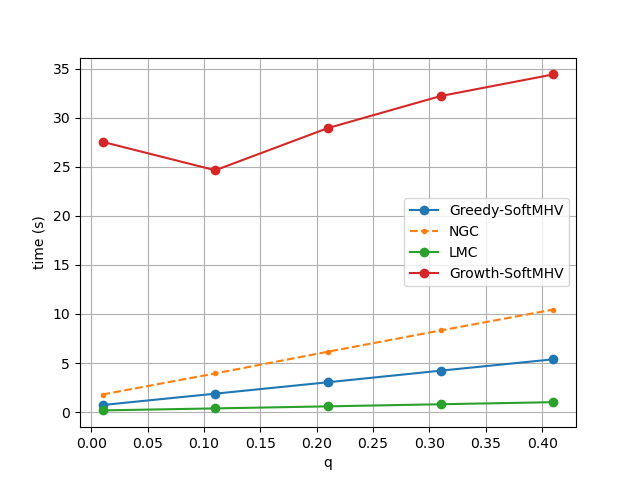} 
\caption{}\label{fig:q-time}
\end{subfigure} 
\hspace{2mm}
\begin{subfigure}{0.5\textwidth}       
\includegraphics[scale=0.5]{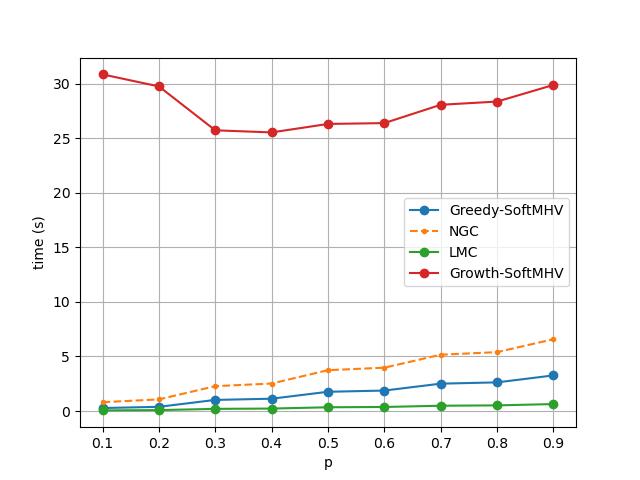}
\caption{}\label{fig:p-time}
\end{subfigure}

\centering
\begin{subfigure}{0.5\textwidth}  
\includegraphics[scale=0.5]{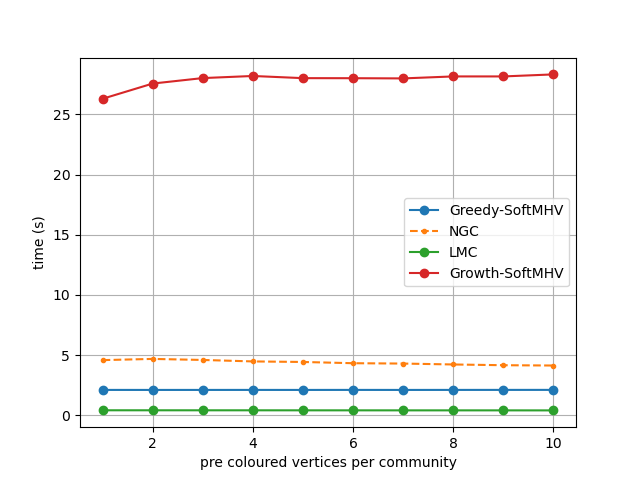} 
\caption{}\label{fig:pcc-time}
\end{subfigure} 
\caption{Comparing the average running time of the four algorithms concerning the parameters (a) $k$, (b) $\rho$, (c) $q$, (d) $p$ and (e) the number of precoloured vertices per community.}
\label{fig:time}
\end{figure}

Interestingly, both algorithms \hyperref[greedy1]{\sf Greedy-SoftMHV} and  \hyperref[greedy2]{\sf NGC} detect a similar number of $\rho$-happy vertices (see Figure~\ref{fig:happy}, especially Figure~\ref{fig:k-happy}). However, they do not detect communities effectively (see Figure~\ref{fig:comm}, especially Figure~\ref{fig:k-comm}). This is unsurprising since they assign only the one colour to all uncoloured vertices. The run times for \hyperref[greedy1]{\sf Greedy-SoftMHV} is substantially lower than  \hyperref[greedy2]{\sf NGC}, and hence, using \hyperref[greedy2]{\sf NGC} cannot be justified for real-world applications.

\begin{figure}
\captionsetup{size=small}
\begin{subfigure}{0.5\textwidth}
\includegraphics[scale=0.5]{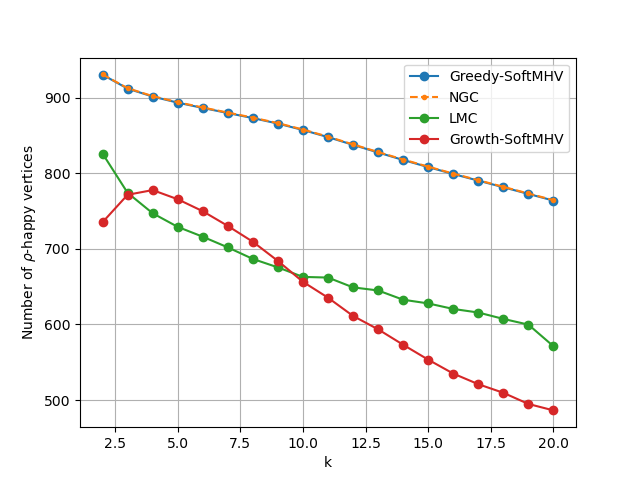} 
\caption{}\label{fig:k-happy}
\end{subfigure} 
\hspace{2mm}
\begin{subfigure}{0.5\textwidth}          
\includegraphics[scale=0.5]{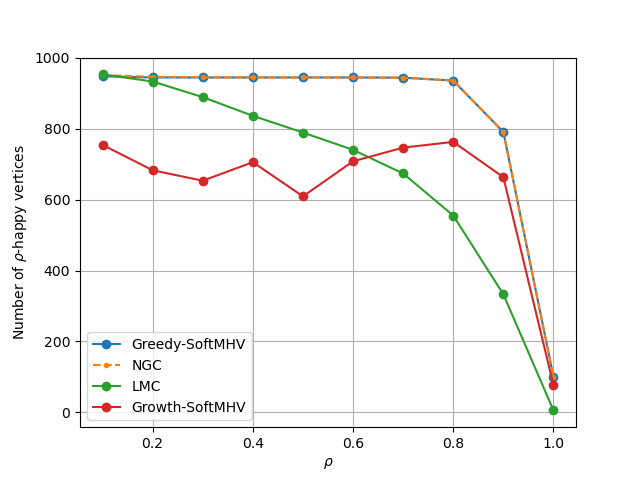}  
\caption{}\label{fig:r-happy}
\end{subfigure} 

\begin{subfigure}{0.5\textwidth}  
\includegraphics[scale=0.5]{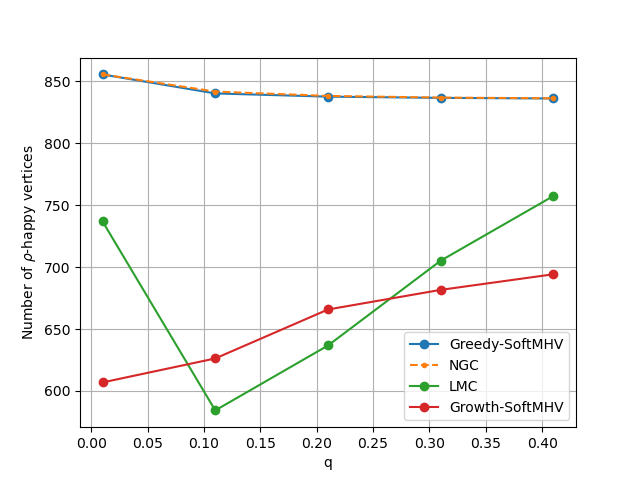} 
\caption{}\label{fig:q-happy}
\end{subfigure} 
\hspace{2mm}
\begin{subfigure}{0.5\textwidth}  
\includegraphics[scale=0.5]{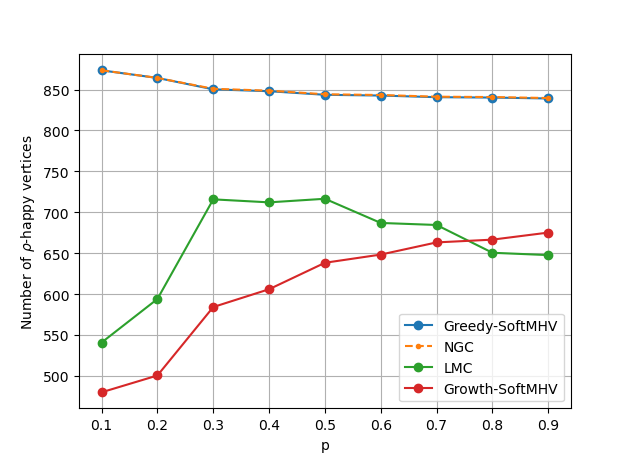}
\caption{}\label{fig:p-happy}
\end{subfigure} 

\centering
\begin{subfigure}{0.5\textwidth}  
\includegraphics[scale=0.5]{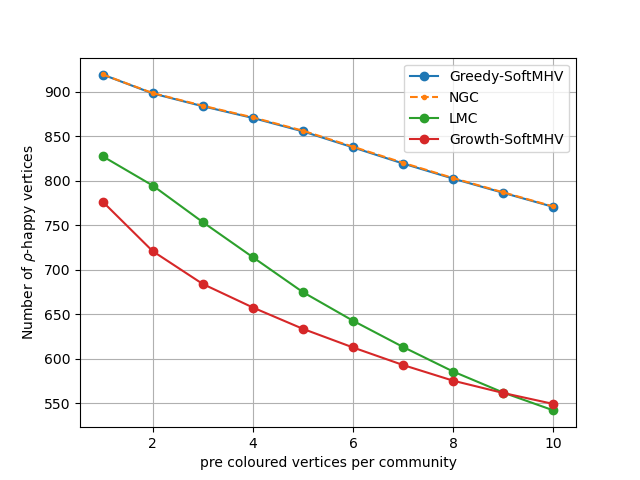} 
\caption{}\label{fig:pcc-happy}
\end{subfigure} 
\caption{Comparing the average number of $\rho$-happy vertices in the output of the four algorithms to the parameters (a) $k$, (b) $\rho$, (c) $q$, (d) $p$ and (e) the number of precoloured vertices per community.}
\label{fig:happy}
\end{figure}

\hyperref[growth]{\sf Growth-SoftMHV} requires large run times, but to make fair comparisons, we limit its CPU run time to 40 seconds (see Figure~\ref{fig:time}). In this setting, all other algorithms find a higher number of $\rho$-happy colourings and typically detect communities more accurately (see Figures~\ref{fig:happy} and~\ref{fig:comm}, especially sub-figures~\ref{fig:p-happy} and~\ref{fig:q-comm}).

\begin{figure}
\captionsetup{size=small}
\begin{subfigure}{0.5\textwidth}
\includegraphics[scale=0.5]{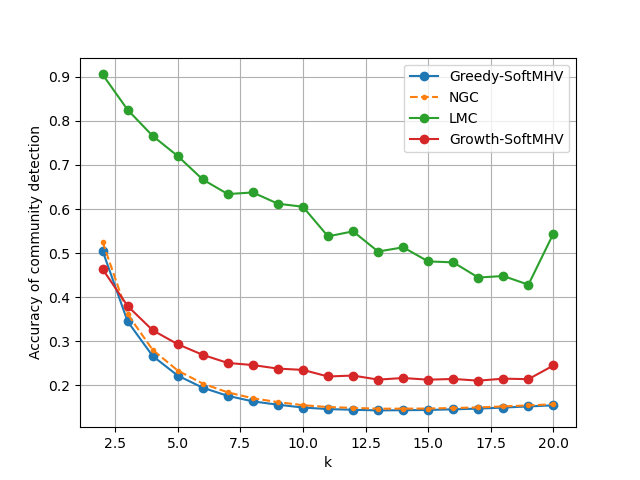} 
\caption{}\label{fig:k-comm}
\end{subfigure} 
\hspace{2mm}
\begin{subfigure}{0.5\textwidth}  
\includegraphics[scale=0.5]{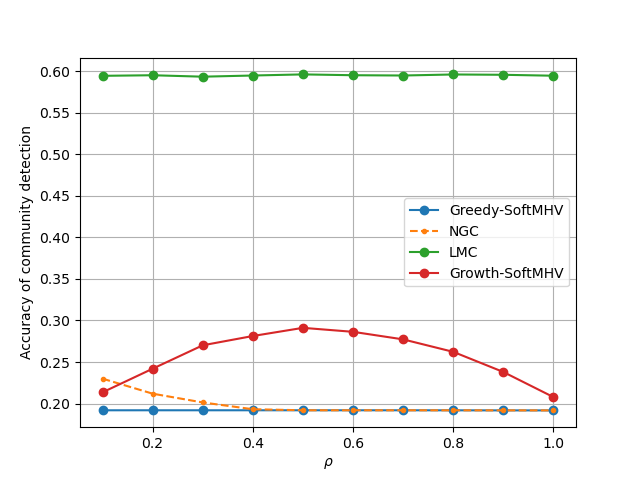}
\caption{}\label{fig:r-comm}
\end{subfigure} 

\begin{subfigure}{0.5\textwidth}  
\includegraphics[scale=0.5]{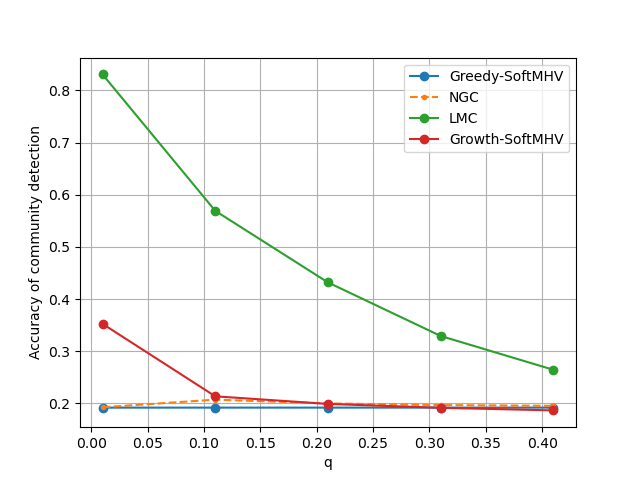} 
\caption{}\label{fig:q-comm}
\end{subfigure} 
\hspace{2mm}
\begin{subfigure}{0.5\textwidth}          
\includegraphics[scale=0.5]{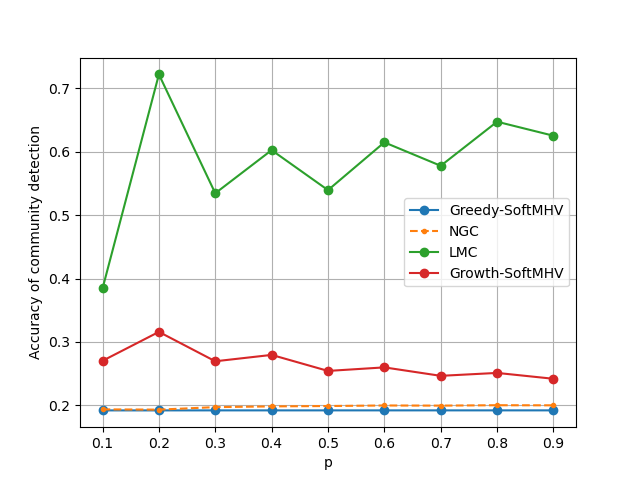}
\caption{}\label{fig:p-comm}
\end{subfigure} 

\centering
\begin{subfigure}{0.5\textwidth}  
\includegraphics[scale=0.5]{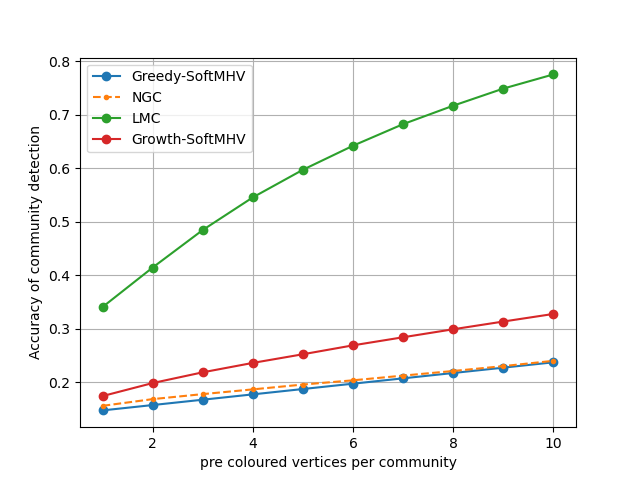} 
\caption{}\label{fig:pcc-comm}
\end{subfigure} 
\caption{Comparing the average quality of community detection of the four algorithms to changes of the parameters (a) $k$, (b) $\rho$, (c) $q$, (d) $p$ and (e) the number of precoloured vertices per community. The quality of community detection is the ratio of vertices whose colours in the output graph agree with their communities in the generated graph.}
\label{fig:comm}
\end{figure}

\hyperref[greedy3]{\sf LMC} has the lowest run times. It does not detect $\rho$-happy vertices as effectively as \hyperref[greedy1]{\sf Greedy-SoftMHV} and  \hyperref[greedy2]{\sf NGC}, but detects communities more accurately than the other algorithms (Figures~\ref{fig:happy} and~\ref{fig:comm}). For nearly 10\% of the tests, it found a colouring where all the vertices are $\rho$-happy, far greater than 1.41\% of the \hyperref[growth]{\sf Growth-SoftMHV} algorithm (see Figure~\ref{fig:bar1-happy} and Table~\ref{table:num-per})\footnote{It should be noted that \hyperref[growth]{\sf Growth-SoftMHV} occasionally reached its time limit of 40 seconds, and only reports its partial solutions. 
Hence, we may observe improved performance if the time limit is not restricted. 
}. For \hyperref[greedy1]{\sf Greedy-SoftMHV} and \hyperref[greedy2]{\sf NGC}, the percentage of complete $\rho$-happy colouring of their outputs are 
inconsiderable, although \hyperref[greedy2]{\sf NGC} outperforms \hyperref[greedy1]{\sf Greedy-SoftMHV}. 
For certain instances, only \hyperref[greedy3]{\sf LMC} detects communities accurately (Figure~\ref{fig:bar1-comm} and Table~\ref{table:num-per}), despite not being an explicit objective.  

\begin{figure}
\centering
\captionsetup{size=small}
\begin{subfigure}{0.9\textwidth}
\includegraphics[scale=0.6]{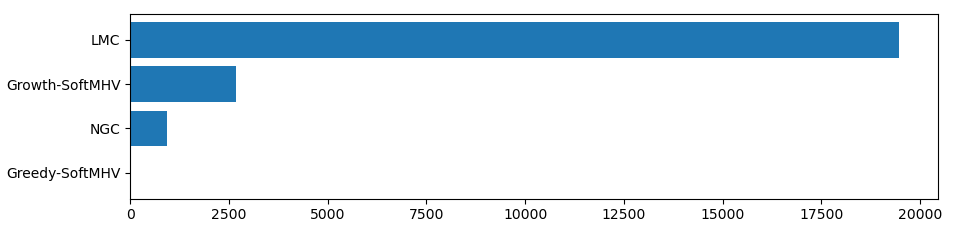}
\caption{}
\label{fig:bar1-happy}
\end{subfigure} 

\begin{subfigure}{0.9\textwidth}  
\includegraphics[scale=0.6]{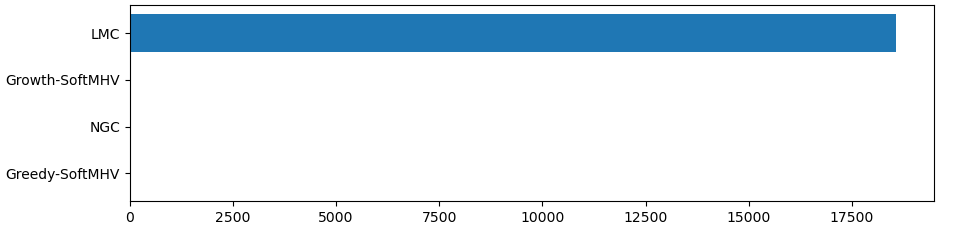}
\caption{}
\label{fig:bar1-comm}
\end{subfigure} 

\caption{(a) The number of times each algorithm generated a colouring that is $\rho$-happy out of 190,000 runs. It can be seen that {\sf LMC} outperforms other algorithms in this regard. \\(b) The number of times that each algorithm detects communities accurately.
}\label{fig:bar}
\end{figure}

\begin{table}%[!ht]
\centering\small
\begin {tabular}{|c||c|c|c|c|} \hline
{\bf Algorithm} & {\bf \# CH$^\ast$} & {\bf \% CH$^\S$} & {\bf\# ACD$^\dagger$} & {\bf\% ACD$^{\ddagger}$} \\ \hline \hline
\hyperref[greedy1]{\sf Greedy-SoftMHV} & 6 & 0.00& 0 & 0.00 \\ \hline
\hyperref[greedy2]{\sf Neighbour Greedy Colouring} & 921 & 0.48 & 0& 0.00 \\ \hline
\hyperref[greedy3]{\sf Local Maximal Colouring}& 19,475 & 10.25 & 18,562 & 9.77 \\ \hline
\hyperref[growth]{\sf Growth-SoftMHV} & 2,688 & 1.41 & 3 & 0.00 \\ \hline
Communities induced colouring & 27,510 & 14.48 & 190,000& 100.00 \\ \hline
\end {tabular}

\medskip
\captionsetup{singlelinecheck=off}
\caption[foo bar]{The number of times and their percentage that the algorithms find colourings that make all the vertices $\rho$-happy versus the number of times and their percentage that they accurately detect communities, out of 190,000 runs. 
\begin{itemize}
\item[$\ast$.] \# CH = ``the number of times it finds a complete $\rho$-happy colouring''
\item[$\S$.] \% CH = ''the percentage of \# CH across 190,000 runs''
\item[$\dagger$.] \# ACD = ``the number of times it accurately detects communities''
\item[$\ddagger.$] \% ACD = ``the percentage of \# ACD across 190,000 runs''
\end{itemize}}\label{table:num-per}
\end{table}

Table~\ref{table:num-inc} shows that \hyperref[greedy3]{\sf LMC} finds the highest number of $\rho$-happy vertices, 987 out of 1,000, when communities induce $\rho$-happy colouring, i.e. when $\rho<\xi$. The best average number of $\rho$-happy vertices across all the 190,000 runs belong to the \hyperref[greedy2]{\sf NGC} with 845 $\rho$-happy vertices. The \hyperref[greedy3]{\sf LMC} has the best accuracy of community detection whether communities induce $\rho$-happy colouring or not.

\begin{table}%[!ht]
\centering\small
\begin {tabular}{|c||c|c|c|c|c|c|} \hline
{\bf Algorithm} & {\bf H$^\ast$} & {\bf AC$^\S$} & {\bf NAC$^\dagger$} & {\bf C$^\heartsuit$} & {\bf ACD$^\diamondsuit$} & {\bf NACD$^{\varclubsuit}$} \\ \hline \hline
\hyperref[greedy1]{\sf Greedy-SoftMHV} & 844 & 973  & 822 & 0.192 & 0.289& 0.175 \\ \hline
\hyperref[greedy2]{\sf Neighbour Greedy Colouring} & 845 & 975  & 823 & 0.199 & 0.313& 0.179 \\ \hline
\hyperref[greedy3]{\sf Local Maximal Colouring}& 671 & 987  & 617 & 0.594 & 0.840& 0.553 \\ \hline
\hyperref[growth]{\sf Growth-SoftMHV} & 636 & 794  & 610 & 0.257 & 0.370& 0.238 \\ \hline
%Communities induced colouring & 27,510 & 14.48 & &190,000& 100.00 \\ \hline
\end {tabular}

\medskip
\captionsetup{singlelinecheck=off}
\caption[foo bar]{Average number of $\rho$-happy vertices and average accuracy of community detection for each algorithm are reported, in the entire test, when communities induce a $\rho$-happy colouring and when they do not. 
\begin{itemize}
\item[$\ast$.]  H = ``Average number of $\rho$-happy vertices across 190,000 runs''
\item[$\S$.] AC = ``Average number of $\rho$-happy vertices when communities induce a complete $\rho$-happy colouring ($\rho<\xi$)''
\item[$\dagger$.] NAC = ``Average number of $\rho$-happy vertices when communities do not induce a complete $\rho$-happy colouring ($\rho\geq\xi$)''
\item[$\heartsuit$.] C = ``Average accuracy of community detection across 190,000 runs''
\item[$\diamondsuit$.] ACD = ``Average accuracy of community detection when communities induce a complete $\rho$-happy colouring ($\rho<\xi$)''
\item[$\varclubsuit.$]  NACD = ``Average accuracy of community detection when communities do not induce a complete $\rho$-happy colouring ($\rho\geq\xi$)''
\end{itemize}}\label{table:num-inc}
\end{table}

\subsection{Test 2: graphs on different numbers of vertices}\label{sec:scale}

The four algorithms in Section~\ref{sec:alg}, especially \hyperref[growth]{\sf Growth-SoftMHV}, are not designed to be run for large graphs. Therefore, if we are interested to see their scalability, we need to limit the number of total runs. For this, another test was necessary. We remind the reader again that we ran the algorithms over randomly generated graphs on $200\leq n<$ 3,000 vertices for this test. Other parameters are randomly chosen over the same interval of the previous test (with an insist on $q\leq \frac{p}{2}$ to make the communities distinguishable). For each $n$ we have generated 10 instances. The results of 28,000 runs of the algorithms are summarised in the following paragraphs and figures.

First, we note that the data for this test agrees with the former test. For example, the average ratio of $\rho$-happy vertices in the induced colouring by the communities is given in Figure~\ref{fig:r-k-comm-new}, which is consistent with Figure~\ref{fig:k-r-happy}, especially when $\rho$ and $k$ are small enough, the induced colouring by the communities can make the entire vertex set $\rho$-happy.

\begin{figure}
    \centering
    \includegraphics[scale=0.6]{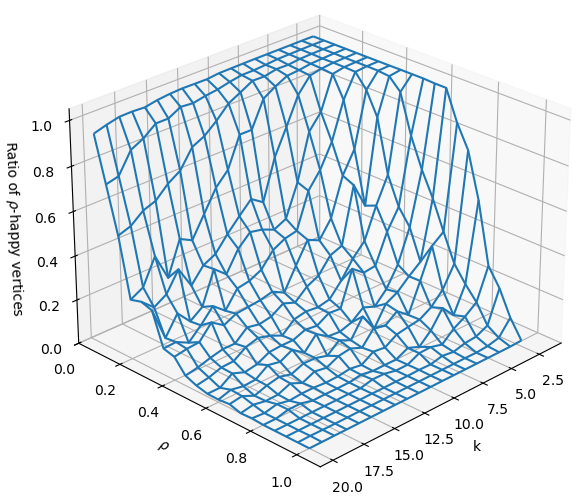}
    \caption{The 3D chart of the average ratio of $\rho$-happy vertices in the induced colouring by communities, with $\rho$ and $k$ as variables. The data is from the second test, where the number of vertices is between 200 and 3,000.}
    \label{fig:r-k-comm-new}
\end{figure}

The time limit for this part is 120 seconds. The only algorithm that occasionally reports incomplete solutions due to this time limit is the \hyperref[growth]{\sf Growth-SoftMHV}, which experiences the sharpest increase in its running time. As expected, increasing the number of vertices directly relates to the running time of the algorithms, which is obvious from Figure~\ref{fig:n-time}. At the same time, \hyperref[greedy3]{\sf LMC} shows a small increase in its running time in comparison to the other algorithms. 

Like the former test for graphs on 1,000 vertices, \hyperref[greedy1]{\sf Greedy-SoftMHV} and \hyperref[greedy2]{\sf NGC} perform almost similarly in finding $\rho$-happy vertices, see Figure~\ref{fig:n-happy}. The average number of $\rho$-happy vertices found by the \hyperref[growth]{\sf Growth-SoftMHV} divided by the total number of vertices drops as the number of vertices increases. This trend for \hyperref[greedy3]{\sf LMC}, which has the highest correlation with communities, remains increasing as Figure~\ref{fig:n-happy} affirms. This complies with Theorem~\ref{th:infinity} that communities can find more $\rho$-happy vertices when $n$ becomes larger.

Accuracy of community detection of \hyperref[greedy1]{\sf Greedy-SoftMHV}, \hyperref[greedy2]{\sf NGC} and \hyperref[growth]{\sf Growth-SoftMHV} drop sharply, according to Figure~\ref{fig:n-comm}. However, \hyperref[greedy3]{\sf LMC} maintains its community detection accuracy between 0.4 and 0.5. 

\begin{figure}

\captionsetup{size=small}
\begin{subfigure}{0.45\textwidth}
\includegraphics[scale=0.42]{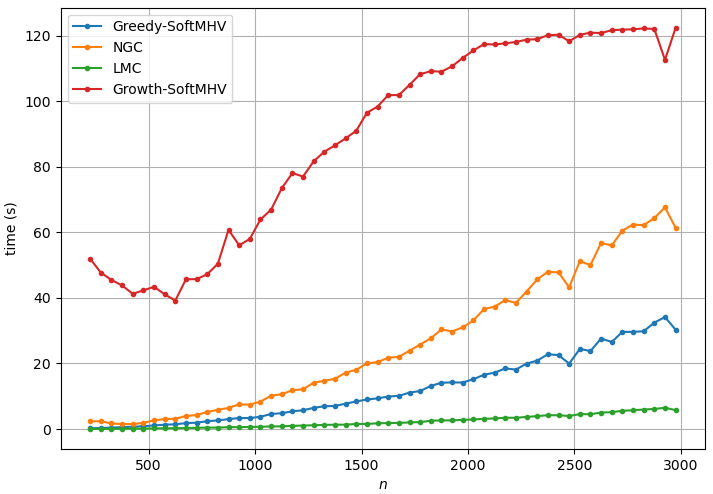} 
\caption{}\label{fig:n-time}
\end{subfigure} 
\hspace{2mm}
\begin{subfigure}{0.5\textwidth}  
\includegraphics[scale=0.42]{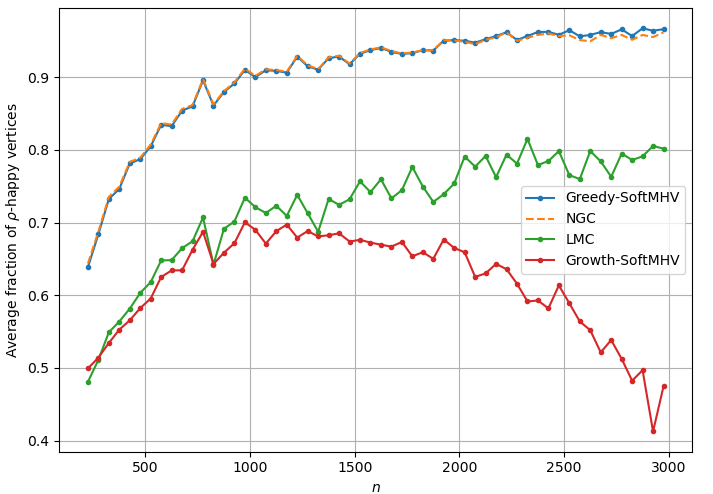}
\caption{}\label{fig:n-happy}
\end{subfigure}

\centering
\begin{subfigure}{0.5\textwidth}  
\includegraphics[scale=0.45]{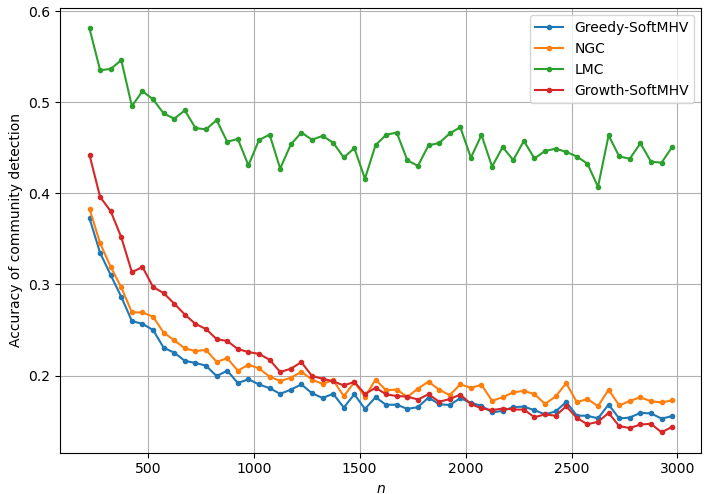} 
\caption{}\label{fig:n-comm}
\end{subfigure} 
\caption{Comparing (a) running time, (b) average number of $\rho$-happy vertices and (c) average accuracy of community detection of the four algorithms. The data is from the second test whose graphs are on $200\leq n<$ 3,000 vertices.}
\label{fig:n-v}
\end{figure}

In general, according to Figure~\ref{fig:n-v}, \hyperref[greedy3]{\sf LMC} is the most affordable algorithm for finding $\rho$-happy vertices, especially when another objective is community detection. In contrast, \hyperref[growth]{\sf Growth-SoftMHV} is costly and may perform poorly when the number of vertices is large.

\subsection{Performance of \hyperref[greedy3]{\sf LMC}}
In this section, we briefly investigate the performance of \hyperref[greedy3]{\sf LMC}. As it can be seen in Line 5 of Algorithm~\ref{greedy3}, an uncoloured vertex is randomly chosen from the neighbourhood of already coloured vertices. It is possible that choosing another vertex at a step affects the quality of the final solution. To investigate this, we performed \hyperref[greedy3]{\sf LMC} 10 times over each of the 28,000 graphs of the second test with $200\leq n<$ 3,000 vertices, introduced in Section~\ref{sec:scale}. The result of running \hyperref[greedy3]{\sf LMC} is reported in Figure~\ref{fig:LMC}.

\begin{figure}

\captionsetup{size=small}
\begin{subfigure}{0.5\textwidth}
\includegraphics[scale=0.42]{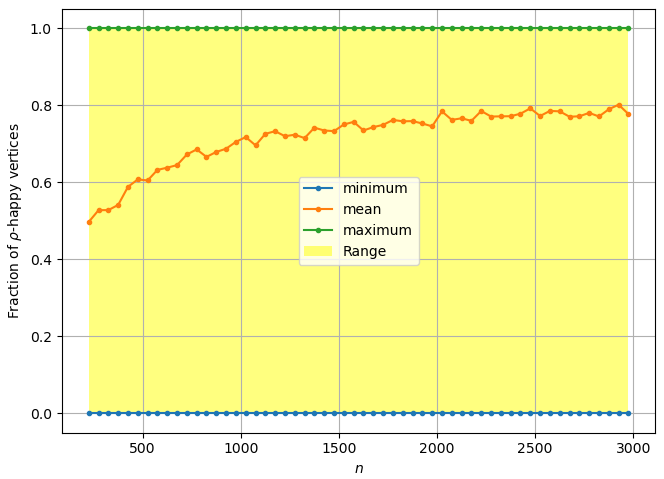} 
\caption{}\label{fig:lmc-general}
\end{subfigure} 
\hspace{2mm}
\begin{subfigure}{0.5\textwidth}  
\includegraphics[scale=0.42]{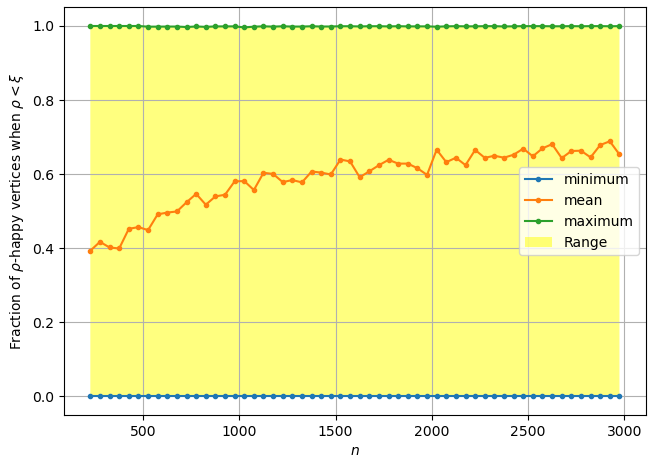}
\caption{}\label{fig:LMC-greater-xi}
\end{subfigure} 

\begin{subfigure}{0.5\textwidth}
\includegraphics[scale=0.42]{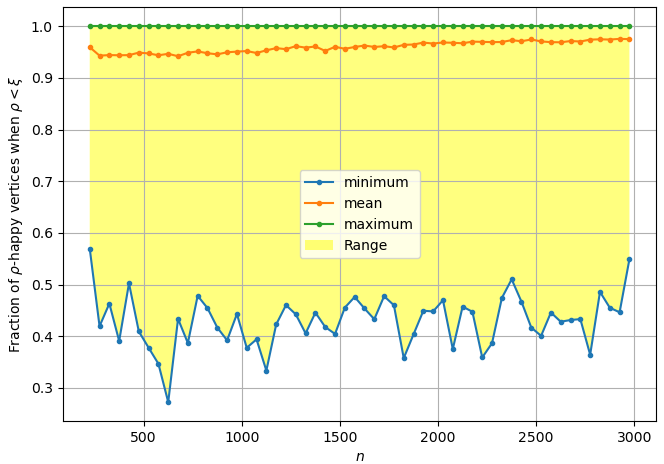} 
\caption{}\label{fig:LMC-less-xi}
\end{subfigure} 
\hspace{2mm}
\begin{subfigure}{0.5\textwidth}  
\includegraphics[scale=0.42]{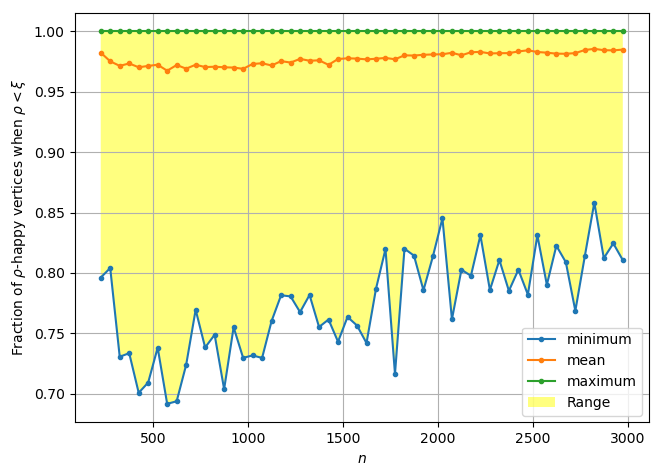}
\caption{}\label{fig:LMC-less-xi2}
\end{subfigure}

\caption{The performance of \hyperref[greedy3]{\sf LMC} (a) when there is no constraints on $\rho$, (b) when $\rho>\xi$, (c) when $\rho<\xi$ and (d) when $\rho<\frac{\xi}{2}$. The data were collected by running \hyperref[greedy3]{\sf LMC} 10 times for each of 28,000 randomly generated graphs with $200\leq n<$ 3,000 vertice of the second test.}
\label{fig:LMC}
\end{figure}

Figure~\ref{fig:lmc-general} shows the minimum, mean and maximum of the number of $\rho$-happy vertices divided by the number of vertices. In this case, the minimum is almost always 0 while the maximum is almost always 1. The mean, however, is between 0.5 and 0.8, and as $n$ increases, the mean also shows a tendency to rise. Figures~\ref{fig:LMC-greater-xi}, \ref{fig:LMC-less-xi} and \ref{fig:LMC-less-xi2} are for the same fraction with additional restrictions of $\rho>\xi$, $\rho<\xi$ and $\rho<\frac{\xi}{2}$, respectively. When $\rho>\xi$, the maximum and minimum do not change, but the mean falls between 0.4 and 0.7. When $\rho<\xi$ (and when $\rho<\frac{\xi}{2}$), the maximum is again 1, but the minimum increases to around 0.4 (resp. 0.75) and the mean increases to around 0.95 (resp. 0.98), which again tends to be higher as $n$ increases.

\section{Conclusion}

This paper theoretically and experimentally establishes a connection between the soft happy colouring of graphs and their community structure. Since conventional happy colouring is a special case of soft happy colouring, and the main intention of defining these sorts of colouring was homophily (which is a property of social networks explainable by community structures), soft happy colouring is a more powerful tool to achieve the intended goal.

One of the achievements of this paper is to theoretically show that for the induced colouring by communities of a graph in the SBM, almost surely all the vertices are $\rho$-happy if $\rho$ is small enough. Another is finding the threshold $\xi$ which indicates how small this $\rho$ should be. The other is showing that if $\rho<\lim_{n\rightarrow \infty} \xi$, then the communities almost certainly induce a $\rho$-happy colouring on the vertices. Experimental analysis of over large number of randomly generated graphs perfectly verified the theoretical expectations. 

Additionally, we developed algorithms to tackle the soft happy colouring problem. We adopt two existing algorithms from \cite{ZHANG2015117} and present a novel algorithm, {\sf Local Maximal Colouring}. The algorithm has no reliance on $\rho$, however, its output shows considerable correlation with the communities of the tested networks while it has an inexpensive running cost. Therefore, for practical purposes, multiple runs of LMC are conducted to find the best possible result for the both objectives of soft happy colouring and community detection. 

There are numerous possibilities for future work. Further studies can include considering the relations of (soft) happy colouring with communities of random graph models other than the simplified stochastic block model $\mathcal{G}(n,k,p,q)$. Looking at the induced colouring of detected communities by algorithms, such as spectral clustering, for their number of $\rho$-happy vertices can also be interesting on its own because these algorithms may find communities different from what graph generators produce. Another area for future work can be maximizing the number of $\rho$-happy coloured vertices using mathematical/constraint programming, metaheuristics (tabu search and evolutionary algorithms), and matheuristics (CMSA and CMSA-TS) for $\rho$-happy colouring similar to~\cite{THIRUVADY2022101188} by Thiruvady and Lewis. Improving known algorithms for soft happy colouring, similar to what Ghirardi and Salassa~\cite{Ghirardi2022-lb} suggest, is another possible study in future. Especially, one can think of improvements to the \hyperref[greedy3]{\sf LMC} because of its low time cost and its output's correlation with the community structure of its input.

\section*{Declarations of interest} 
None

\bibliographystyle{plain}
\bibliography{HC.bib}
\end{document}